\documentclass{elsarticle}

\usepackage{a4wide}

\usepackage{amsthm}
\usepackage{amsmath,amssymb}
\usepackage{amsfonts}

\usepackage{tikz}
\usetikzlibrary{decorations.pathreplacing}
\usetikzlibrary{shapes}

\newtheorem{theorem}{Theorem}
\newtheorem{corollary}{Corollary}
\newtheorem{lemma}{Lemma}
\newtheorem{fact}{Fact}
\newtheorem{assumption}{Assumption}
\newtheorem{observation}{Observation}

\newcommand{\eps}{\varepsilon}

\newcommand{\OO}{\ensuremath{{O}}}

\newcommand{\BPA}{\textsc{bpa}}
\newcommand{\BPAA}{\textsc{abpa}}

\newcommand{\OPT}{\textsc{opt}}
\newcommand{\Sstar}{S^*}

\newcommand{\Sprime}{S'}

\newcommand{\Salg}{S}
\newcommand{\ALG}{\textsc{alg}}
\newcommand{\ADV}{\textsc{adv}}

\newcommand{\bpAdvUp}{\frac{1}{\eps}\log\left(\frac{2}{\eps^2}\right) + \log\left(\frac{2}{\eps^2}\right)  + 3}
\newcommand{\msTypes}{\left\lceil \log_{1+\eps}\frac1\eps \right\rceil}

\newcommand{\schdBigOAdv}{\OO\left(\frac{1}{\eps}\log \frac1\eps \right)}

\newcommand{\half}[1]{\frac{1}{2^{#1}}}

\newcommand{\bpNull}{\bot}

\newcommand{\sAdvice}{\begin{center}\begin{tabular}{l r p{0.6\textwidth}}}
\newcommand{\fAdvice}{\end{tabular}\end{center}}

\begin{document}

\begin{frontmatter}

\title{Online Algorithms with Advice for \\ Bin Packing and Scheduling Problems}

\author[inst1]{Marc P. Renault\fnref{curMarc,netoc}}
\ead{marc.renault@lip6.fr}

\author[inst1]{Adi Ros\'{e}n\fnref{netoc}}
\ead{adiro@liafa.univ-paris-diderot.fr}
\address[inst1]{CNRS and Universit\'e Paris Diderot, France}

\author[inst3]{Rob van Stee}
\ead{rob.vanstee@leicester.ac.uk}
\address[inst3]{University of Leicester, Department of Computer Science, University Road, Leicester, LE1 7RH, United Kingdom}

\fntext[curMarc]{Present address: Sorbonne Universit\'{e}s, UPMC Univ Paris 06, UMR 7606, LIP6, F-75005, Paris, France} 
\fntext[netoc]{Research supported in part by ANR project NeTOC.}   

\begin{abstract}
We consider the setting of online computation with advice and study the bin packing problem
and a number of scheduling problems. 
We show that it is possible, for any of these
problems, to arbitrarily  approach a competitive ratio of $1$ with only a constant
number of bits of advice per request. For the bin packing problem, we give an online algorithm
with advice that is
 $(1+\eps)$-competitive and uses $\OO\left(\frac{1}{\eps}\log \frac{1}{\eps} \right)$ bits of advice  per request.
For scheduling on $m$ identical machines, with the objective function of any of makespan, machine covering
and the minimization of the $\ell_p$ norm, $p >1$, we give similar results. We give online algorithms
with advice which are $(1+\eps)$-competitive ($(1/(1-\eps))$-competitive for machine covering)  and also use
$\OO\left(\frac{1}{\eps}\log \frac{1}{\eps} \right)$ bits of advice per request. We complement our results by giving a lower bound that shows
that for any online algorithm with advice to be optimal, for any of the above scheduling problems,
a non-constant number (namely, at least $\left(1 - \frac{2m}{n}\right)\log m$, where $n$ is the number of jobs and $m$ is the number of machines) of bits of advice per request is needed.
\end{abstract}


\end{frontmatter}

\section{Introduction}

Online algorithms are algorithms that receive their input one piece at a time.  An online algorithm must make an irreversible decision on the processing of the current piece of the input before it receives the next piece, incurring a cost for this processing. The method of choice to analyze such algorithms is {\em competitive analysis}~\cite{BorodinBook1998}. In this framework, the decisions of the online algorithm must be taken with no knowledge about future pieces of input.  In competitive analysis, one measures the quality of an online algorithm by analyzing its {\em competitive ratio}, i.e.\ the worst-case ratio, over all possible finite request sequences, of the cost of the online algorithm and the cost of an optimal offline algorithm that has full knowledge of the future. 
In general, there are no computational assumptions made about the online algorithm, and thus competitive analysis is concerned with quantifying the difference between no knowledge of the future and full knowledge of the future.

In many situations, however, an algorithm with no knowledge of the future is unreasonably restrictive \cite{BorodinBook1998,DorrigivL2005}.  Furthermore, ``classical'' competitive analysis, as described above, is only concerned with one point on the spectrum of the amount of information
about the future available to the online algorithm (i.e.\ no information at all). In order both to address the lack of a general model for situations of partial information about the future, and to try to quantify the interplay between
the amount of information about the future and the achievable competitive ratio, a framework for a more refined analysis
of online algorithms, which attempts to analyze online algorithms with partial information about the future, has been proposed and studied in recent years, e.g. ~\cite{BockenhauerKKKM09,EmekFKR11,BockenhauerKKK11,KommK11,BockKKR14,RenaultR2015,DorrigivHZ2012,BoyarKLL14}. 
 
This framework was dubbed {\em online computation with advice} and, roughly speaking
(see Section~\ref{pre:advice} for a formal definition), works as follows. The online algorithm, when 
receiving each piece of input $r_i$, can query the adversary about the future by specifying 
a function $u_i$ going from the universe of all input sequences to a universe of all binary strings
of length $b$, for some $b \geq 0$. The adversary must respond with the value of the function 
on the whole input sequence (including the parts not yet revealed to the online algorithm). 
Thus, the online algorithm receives, with each piece of input, $b$ bits of information about 
the future. We call these {\em bits of advice}.  The decisions of the online algorithm can 
now depend not only on the input seen so far, but also on the advice bits received so far which 
reveal some information about the future. 
The online algorithm can thus improve its competitive ratio. 
We are typically interested in the interplay between
the amount of information received about the future 
and the achievable competitive ratio. This model was introduced by Emek et al.~\cite{EmekFKR11}. Another variant of the setting of online algorithms with advice was proposed by 
 B{\"o}ckenhauer et al.~\cite{BockenhauerKKKM09} (see Section~\ref{intro:rw}). Recent years have seen 
 an emergence of works on online computation with advice in both variants of the model, e.g. 
 studying problems such as the $k$-server problem~\cite{EmekFKR11,BockenhauerKKK11,RenaultR2015,GuptaKL13}, metrical task systems~\cite{EmekFKR11}, the knapsack problem~\cite{BockKKR14}, the bin packing problem~\cite{BoyarKLL14}, 2 value buffer management~\cite{DorrigivHZ2012}, reordering buffer management problem~\cite{AdamaszekRRvS2013} 
and more. 
 
In this paper, we study  bin packing, and scheduling on $m$ identical machines with the objective 
functions of the makespan, machine covering, and minimizing the $\ell_p$ norm
in the framework of online computation with advice.  All of these problems
have been widely studied in the framework of online algorithms (without advice), and 
in the framework of offline approximation algorithms, e.g.~\cite{VegaL1981,HochbaumS1987,ChenVW1994,Alon1997,Sgall1997,Woeginger1997,AvidorAS1998,AzarE1998,FleischerW2000,Seiden2001,Albers2002,RudinC2003,BaloghBG2012}. For all of these
problems, we show that it is possible to arbitrarily approach a competitive ratio of $1$ with a
constant number of bits of advice per request, i.e.\ we give $(1+\eps)$-competitive 
deterministic algorithms
with advice that use $f(1/\eps)$ bits of advice per request (for some polynomial function $f$).
It is worthwhile noting that this is certainly not the case for all online problems, as 
non-constant lower bounds on the amount of advice required to have a competitive ratio arbitrarily close to $1$ are known for some online problems
(e.g. for metrical task systems~\cite{EmekFKR11}). Furthermore, for all the problems we study, lower bounds
bounded away from $1$ are known for the competitive ratio achievable by online
algorithms without advice.  We further show, for the scheduling problems, that a non-constant number of bits of advice is needed for an online 
algorithm with advice to be optimal (a similar result for bin packing has been given
in~\cite{BoyarKLL14}). 
 
\subsection{Related Work.}\label{intro:rw}
 The model of online computation with advice that we consider in the present paper was
introduced by Emek et al~\cite{EmekFKR11}. 
In the model of~\cite{EmekFKR11}, the advice is a fixed amount that is revealed in an online manner with every request. This model is referred to as the \emph{online advice model}.
Another variant of the model of online algorithms with advice
was proposed by B{\"o}ckenhauer et al.~\cite{BockenhauerKKKM09}. In this variant, the 
advice is not given to the algorithm piece by piece with each request, but rather a single
tape of advice bits is provided to the algorithm.  This model is termed the \emph{semi-online advice model} since the algorithm can read from the advice tape at will and, therefore, could read all the advice at the beginning prior to receiving any requests. For the semi-online advice model, one then 
analyzes the total number of advice bits read from the tape as a function of the length of the
input (and the competitive ratio).  
A number of works have analyzed various online 
problems in the framework of online algorithms with advice (in both variants).
For example: the $k$-server problem has a competitive ratio of at most $\Big\lceil\frac{\lceil\log k\rceil}{b-2}\Big\rceil$~\cite{RenaultR2015}, where $b$ is the number of bits of advice per request; the metrical task system problem has tight bounds on the competitive ratio of $\Theta(\log N / b)$~\cite{EmekFKR11}; the unweighted knapsack problem has a competitive ratio of $2$ with $1$ bit of advice in total and $\Omega(\log(n))$ bits are required to further improve the competitive ratio~\cite{BockKKR14},
 the 2 value buffer management problem has a competitive ratio of $1$ with  $\Theta((n/B)\log B)$ bits of advice ($n$ is the length of the request sequence and $B$ is the size of the buffer)~\cite{DorrigivHZ2012}, and the reordering buffer problem, for any $\eps > 0$, has a $(1+\eps)$-competitive algorithm which uses only a constant (depending on $\eps$) number of advice bits per input item~\cite{AdamaszekRRvS2013}.

To the best of our knowledge, the only scheduling problems studied to date in the
framework of online computation with advice is a special case of the job shop scheduling 
problem~\cite{BockenhauerKKKM09,KommK11}
, and, makespan scheduling on identical machines in \cite{Dohrau2015}. In both cases, the semi-online advice model is used. In \cite{Dohrau2015},
an algorithm that is $(1+\eps)$-competitive and uses advice of constant size in total is presented. 
Boyar et al.~\cite{BoyarKLL14}
studied the bin packing problem with advice, using the semi-online advice model of 
~\cite{BockenhauerKKKM09} and presented a $3/2$-competitive algorithm, using $\log n + o(\log n)$ bits of advice in total, and a $(4/3 + \eps)$-competitive algorithm, using $2n + o(n)$ bits of advice in total, where $n$ is the length of the request sequence.  As both algorithms rely on reading $\OO(\log(n))$ bits of advice prior to receiving any requests, they would use $\OO(\log(n))$ bits of advice per request in the model used in this paper.  The $3/2$-competitive algorithm can be converted into an algorithm that uses $1$ bit of advice per request.  We are not aware of a similar simple conversion for the $(4/3 + \eps)$-competitive algorithm. It should be noted that in the online advice model, an algorithm receives at least $1$ bit of advice per request, i.e.\ at least linear advice in total. Finally, they show that an online algorithm with advice requires at least $(n - 2N)\log N$ bits of advice in total to be optimal, where $N$ is the optimal number of bins.
In \cite{AngelopoulosDKRR2015}, an algorithm is presented that has a competitive ratio that can be arbitrarily close to $1.47012$ and uses constant advice in total. 
Further, they show that linear advice in total is required for a competitive ratio better than $7/6$. 
 
For online bin packing without advice, the best known lower bound on the competitive ratio is 1.54037 due to Balogh et al.~\cite{BaloghBG2012} and the best known deterministic upper bound on the competitive ratio is 1.58889 due to Seiden~\cite{Seiden2001}. Chandra~\cite{Chandra92} showed that all known lower bounds can be shown to apply to randomized algorithms.

For online scheduling on $m$ identical machines without advice, Rudin and Chandrasekaran~\cite{RudinC2003} presented the best known deterministic lower bound of $1.88$ on the competitive ratio for minimizing the makespan.
The best known deterministic upper bound on the competitive ratio for minimizing the makespan, due to Fleischer et al.~\cite{FleischerW2000}, is 1.9201 as $m \to \infty$.  The best known randomized lower bound on the competitive ratio for minimizing the makespan is $1/(1 - (1 - 1/m)^m)$, which tends to $e/(e-1) \approx 1.58$ as $m \to \infty$, and it was proved independently by Chen et al.~\cite{ChenVW1994} and Sgall~\cite{Sgall1997}. The best known randomized algorithm, due to Albers~\cite{Albers2002}, has a competitive ratio of 1.916.

For machine covering, Woeginger~\cite{Woeginger1997} proved tight $\Theta(m)$ bounds on the competitive ratio for deterministic algorithms, and Azar and Epstein~\cite{AzarE1998} showed a randomized lower bound of $\Omega(\sqrt{m})$ and a randomized upper bound of $O(\sqrt{m}\log m)$.  
Also, Azar and Epstein considered the case where the optimal value is known to the algorithm and showed that, for $m \ge 4$, no deterministic algorithm can achieve a competitive ratio better than 1.75.

In the offline case, Fernandez de la Vega and Lueker~\cite{VegaL1981} presented an asymptotic polynomial time approximation scheme (APTAS) for the bin packing problem. Hochbaum and Shmoys~\cite{HochbaumS1987} developed a polynomial time approximation scheme (PTAS) for the makespan minimization problem on $m$ identical machines. Subsequently, Woeginger~\cite{Woeginger1997} presented a PTAS for the machine covering problem on $m$ identical machines and Alon et al.~\cite{Alon1997} presented a PTAS for the $\ell_p$ norm minimization problem on $m$ identical machines.

\subsection{Our Results.}
We give a deterministic online algorithm with advice for bin packing that, for $0 < \eps \le 1/2$, achieves a competitive ratio
of $1+\eps$, and uses $\OO\left(\frac{1}{\eps}\log\frac{1}{\eps}\right)$  bits of advice 
per request. For scheduling on $m$ identical machines, we consider the objective
functions of {\em makespan}, {\em machine covering} and {\em minimizing the $\ell_p$ norm} for $p > 1$. 
For any of these, we give online algorithms with 
advice that, for $0 < \eps < 1/2$, are $(1+\eps)$-competitive 
($(1/(1-\eps))$-competitive for machine covering) and use $\schdBigOAdv$ bits of advice per request. 

We complement our results by showing that, for any of the scheduling problems we consider,
an online algorithm with advice needs at least  $\left(1 - \frac{2m}{n}\right)\log m$ bits
of advice per request to be optimal, where $n$ is the number of jobs and $m$ is the number of machines. This lower bound uses techniques similar to those used by the analogous lower bound for bin packing found in \cite{BoyarKLL14}.
We note that with $\lceil \log m \rceil$ bits a trivial algorithm that indicates for each
job on which machine it has to be scheduled is optimal. 

\subsection{Our Techniques.}
Common to all our algorithms for the packing and scheduling problems is the technique of classifying the input items, according to their size, into a constant number of classes, depending on $\eps$. For the bin packing problem, there are a constant number of groups of a constant number of items with both of these constants depending on $\eps$. For the scheduling problems, the sizes of the items in one class differ only by a multiplicative constant factor, depending on $\eps$. We classify all the items except the smallest ones in this way, where the bound on the size of the items not classified again depends on $\eps$. This classification is done explicitly in the scheduling algorithms, and implicitly in the bin packing algorithms. We then consider an optimal packing (resp. schedule) for the input sequence and define {\em patterns} for the bins (the machines) that describe how the critically sized items (jobs) are packed (scheduled) into the bins (machines). The advice bits indicate with each input item into which bin (machine) pattern it should be packed (scheduled). For the bin packing problem, all but the largest classified items can be packed into the optimal number of bins, according to the assigned pattern. The remaining items cause an
 $\eps$ multiplicative increase in the number of bins used. For the scheduling problems, since items in the same class are ``similar'' in size, we can schedule the items such that the class of the items of each machine matches the class of those in the optimal schedule while being within an $\eps$ factor of the optimal. For both the bin packing problem and the scheduling problems, the very small items (jobs) have to be treated separately; in both cases, the items (jobs) are packed (scheduled) while remaining within an $\eps$ multiplicative factor of the optimal. 

Our techniques for these algorithms are similar to those of~\cite{VegaL1981,HochbaumS1987,Woeginger1997,Alon1997}. In
particular, we use the technique of rounding and grouping the items. The main difficulty in getting our algorithms to work stems from the fact that we must encode the necessary information using only a constant number of advice bits per request. In particular, the number of advice bits per request cannot depend on the size of the input or the size of the instance (number of bins/machines). Further, for the online advice mode, the advice is received per request and this presents additional challenges as the advice has to be presented sequentially per request such that the algorithm will be able to schedule the items in an online manner.

The scheduling objective functions that we consider are all a function of the loads of the machines. This relates closely to the bin packing problem. The main differences are that the bins in the bin packing problem have a maximum capacity and the goal is to  minimize the number of bins used. For scheduling on $m$ identical machines, we have no such capacity constraint (i.e.\ there is no maximum load per machine) but can use at most $m$ machines. This changes the nature of the problem and requires similar but different ideas for the approximation schemes for scheduling as compared to the approximation schemes for bin packing. This is also the case for the online algorithms with advice presented in this paper. The difference is most noticeable in the nature of the grouping of the items that are done implicitly in the case of bin packing based on a ranking of the size of the items and explicitly in the case of scheduling based on a threshold value. 

\section{Preliminaries}\label{sec:pre}

Throughout this paper, we denote by $\log$ the logarithm of base $2$. 
For simplicity of presentation, we assume that $1/\eps$ is a natural number.

 \subsection{Online Advice Model.}\label{pre:advice}
   We use the model of online computation with advice introduced in~\cite{EmekFKR11}. 
A deterministic online algorithm with advice is defined by the sequence of pairs $(g_i, u_i)$, $i \geq 1$. 
The functions $u_i : R^* \to U$ are the query functions where $R^*$ is the set of all finite request sequences, and $U$ is an advice space of all binary strings
of length $b$, for some $b \geq 0$. 
For a given request sequence $\sigma \in R^*$, the advice received with each request $r_i \in  \sigma$ is the value of the function $u_i(\sigma)$. The functions $g_i : R^i \times U^i \to A_i$ are the action functions, 
where $A_i$ is the action space of the algorithm at step $i$.
That is, for request $r_j$, the action of the online algorithm with advice is $a_j = g_j(r_1,\ldots,r_j,u_1,\ldots,u_j)$, i.e.\ a function of the requests and advice received to date. 

\subsection{Competitive Analysis.}
Let $\ALG(\sigma)$ be the cost for an online algorithm $\ALG$ to process $\sigma$ and let $\OPT(\sigma)$ be the optimal cost.  For a minimization problem, 
an online algorithm is {\em $c$-competitive} if, for all finite request sequences $\sigma$, $\ALG(\sigma) \le c\cdot\OPT(\sigma) + \zeta$, where $\zeta$ is a constant that does not depend on $\sigma$. For a maximization problem, an algorithm $\ALG$ is $c$-competitive if $\ALG(\sigma) \geq \frac{1}{c}\OPT(\sigma) - \zeta$.

\subsection{Bin Packing.}\label{sec:bp} 
An instance of the {\em online bin packing problem} consists of a request sequence $\sigma$, and an initially empty set $B$ of bins of capacity $1$.  Each $r_i \in \sigma$ is an item with size $0 < s(r_i) \le 1$.  The goal is to assign all the items of $\sigma$ to bins such that, for each bin $b_j \in B$, $\sum_{r_i \in b_j} s(r_i) \le 1$ and $|B|$ is minimized. 
The optimal number of bins ($|B^{OPT}|$) is denoted by $N$.  
An item \emph{fits} into a bin if its size plus the size of previously packed items in that bin is at
most $1$.
For an item $r_i \in b_j$, where $b_j$ is a bin in the packing $B$,
we will write $r_i \in B$.  
In order to define part of the advice used by our  algorithms, we use a common heuristic for bin packing, {\sc next fit} \cite{Johnson1973}. For completeness, we indicate here that 
the heuristic {\sc next fit} packs the item into the current bin if it fits. 
Else, it closes the current bin, opens a new bin and packs the item in it.

\subsection{Scheduling on $m$ Identical Machines.}  
An instance of the {\em online scheduling problem on $m$ identical machines} consists of $m$ identical machines and a request sequence $\sigma$.  Each $r_i \in \sigma$ is a job with a processing time $v(r_i) > 0$.  
An assignment of the jobs to the $m$ machines is called a {\em schedule}.  For a schedule $S$, $L_i(S) = \sum_{r_j \in M_i} v(r_j)$ denotes the {\em load} of machine $i$ in $S$, where $M_i$ is the set of jobs assigned to machine $i$ in $S$.  In this paper, we focus on the following objective functions:
\begin{itemize}
\item {\em Minimizing the makespan:} minimizing the maximum load
  over all the machines;
\item {\em Machine cover:} maximizing the
  minimum load; 
\item {\em $\ell_p$ norm:} minimizing the $\ell_p$ norm,
  $1 < p \le \infty$, of the load of all the machines.  For a schedule
  $S$, the $\ell_p$ norm is defined to be $\|L(S)\|_p =
  \left(\sum_{i=1}^m(L_i(S))^p\right)^{1/p}$.  Note that minimizing
  the $\ell_\infty$ norm is equivalent to minimizing the makespan.
\end{itemize}

\section{Online Algorithms with Advice for Bin Packing}\label{sec:bin}
Presented in this section is an algorithm for the online bin packing problem called $\BPA$. The algorithm $\BPA$ is inspired by the APTAS algorithms for offline bin packing problem and is $(1 + \eps)$-competitive, using $\OO\left(\frac{1}{\eps}\log \frac{1}{\eps} \right)$ bits of advice per request. 

The advice for the algorithm $\BPA$ is based on a $(1+2\eps)$-competitive packing of the request sequence, denoted by $S$. The packing of $S$ is based on the APTAS of Fernandez de la Vega and Lueker~\cite{VegaL1981} but $S$ is created in an online manner so that $\BPA$ can produce the same schedule. All the items with size at least $\eps$ are grouped, based on their size, into $1/\eps^2$ groups. The groups are numbered sequentially and each item is assigned an \emph{item type} that corresponds to its group number. The packing of the these items uses $(1+\eps)N$ bins and items smaller than $\eps$ can be packed using no more than an additional $\eps N$ bins. The advice indicates the item type and the packing of the bin in which the item is packed. The packing of the bin is described by the types of the items in the bin of $S$. This allows $\BPA$ to reproduce the packing of $S$.

For the $(1 + \eps)$-competitive algorithm, the advice is defined based on an optimal packing of the request sequence. That is, the offline oracle must solve an NP-hard problem. This is possible in this model as no computational restrictions are placed on the oracle. However, it should be noted that the algorithm presented here creates a packing that is $(1 + \eps)$-competitive with respect to some packing $\Sstar$ which does not necessarily have to be an optimal packing. If the computational power of the oracle were restricted, a $(1+\eps)$-competitive (asymptotic) algorithm could be achieved by defining the advice based on the $(1+\eps')$-approximate packing $\Sstar$ created via a bin packing APTAS, e.g.\ the scheme of Fernandez de la Vega and Lueker~\cite{VegaL1981}, (albeit requiring slightly more bits of advice as $\eps$ would have to be adjusted according to $\eps'$). 

The main result of this section is the following.
\begin{theorem}\label{thm:bpaCR}
Given  $\eps$, $0<\eps \le 1/2$, the competitive ratio for $\BPA$ is at most $1 + 3\eps$, and $\BPA$ uses at most $\bpAdvUp$ bits of advice per request.
\end{theorem}

Initially, we will present an algorithm, $\BPAA$, that uses less than $\bpAdvUp$ bits of advice per request and is asymptotically (in the number of optimal bins) $(1+2\eps)$-competitive. Then, with a small modification to $\BPAA$, we will present $\BPA$, an algorithm that is $(1+3\eps)$-competitive for any number of optimal bins and uses $1$ more bit per request than $\BPAA$. That is, regardless of the optimal cost, $\BPA$ always has a cost that is at most $(1+3\eps)$ times the number of optimal bins.

\subsection{Asymptotic $(1+2\eps)$-Competitive Algorithm.}

We begin by creating a rounded input $\sigma'$ based on $\sigma$ using the scheme of Fernandez de la Vega and Lueker~\cite{VegaL1981}. That is, we will group items based on their size into a finite number of groups and round the size of all the items of each group up to the size of the largest item in the group (see Figure~\ref{fig:bpRound}).  

An item is called \emph{large} if it has size larger than $\eps$. Items with size at most $\eps$ are called \emph{small} items. 
Let the number of large items in $\sigma$ be $L$. Sort the large items of $\sigma$ in order of nonincreasing size. Let $h=\lceil\eps^2L\rceil$. For $i=0,\dots,1/\eps^2-1$, assign the large items $ih+1,\dots,ih+\lceil\eps^2L\rceil$ to group $i+1$. A large item of \emph{type} $i$ denotes a large item assigned to group $i$. The last group may contain less than $\lceil\eps^2L\rceil$ items. For each item in group $i$, $i=1,\dots,1/\eps^2$, round up its size to the size of the largest element in the group. 

\begin{figure}[htb]
\centering
\begin{tikzpicture}[thick]

\draw[black] (0, 0.2) -- ++(0,2);
\draw[fill] (0.902132177725434, 0.4) circle (2pt);
\draw[fill] (1.45366716664284, 1.2) circle (2pt);
\draw[fill] (2.2599522722885, 0.8) circle (2pt);
\draw[fill] (2.47207537293434, 1.6) circle (2pt);
\draw[black!60] (2.47207537293434, 0.2) -- ++(0,2);
\draw[fill] (2.5510832760483, 1.2) circle (2pt);
\draw[fill] (2.96437476761639, 2) circle (2pt);
\draw[fill] (3.23697155807167, 0.8) circle (2pt);
\draw[fill] (3.49513975903392, 1.6) circle (2pt);
\draw[fill] (3.82415786385536, 0.4) circle (2pt);
\draw[black!60] (3.82415786385536, 0.2) -- ++(0,2);
\draw[fill] (4.32760030496866, 0.4) circle (2pt);
\draw[fill] (4.41621185280383, 2) circle (2pt);
\draw[fill] (5.0322353374213, 1.6) circle (2pt);
\draw[fill] (5.26406505145133, 0.8) circle (2pt);
\draw[fill] (5.53856943268329, 1.2) circle (2pt);
\draw[black!60] (5.53856943268329, 0.2) -- ++(0,2);
\draw[fill] (5.7159958826378, 0.8) circle (2pt);
\draw[fill] (6.12248234916478, 2) circle (2pt);
\draw[fill] (7.18547832686454, 1.2) circle (2pt);
\draw[fill] (7.87348593790084, 1.6) circle (2pt);
\draw[fill] (7.99700583759695, 0.4) circle (2pt);
\draw[black!60] (7.99700583759695, 0.2) -- ++(0,2);
\draw[fill] (8.31166283432394, 2) circle (2pt);
\draw[fill] (8.6076991211623, 1.6) circle (2pt);
\draw[fill] (8.71852759253234, 1.2) circle (2pt);
\draw[fill] (8.71912145968527, 0.4) circle (2pt);
\draw[fill] (9.76098943501711, 0.8) circle (2pt);
\draw[black!60] (9.76098943501711, 0.2) -- ++(0,2);
\draw[black] (10, 0.2) -- ++(0,2);

\node[left] at (-0.2, 2.5) {Type:};
\node at (1.23603768646717, 2.5) {$5$};
\node at (3.14811661839485, 2.5) {$4$};
\node at (4.68136364826933, 2.5) {$3$};
\node at (6.76778763514012, 2.5) {$2$};
\node at (8.87899763630703, 2.5) {$1$};
\node[left] at (-0.2,1.1) {$\sigma$:};
\node[left] at (-0.2,-0.05) {Size:};
\node at (0,-0.05) {$\eps \phantom{1}$};
\node at (10,-0.05) {$1$};
\node[left] at (-0.2,-1.25) {$\sigma'$:};

\draw[fill] (2.47207537293434, -2.1) circle (2pt);
\draw[fill] (2.47207537293434, -1.3) circle (2pt);
\draw[fill] (2.47207537293434, -1.7) circle (2pt);
\draw[fill] (2.47207537293434, -0.9) circle (2pt);
\draw[fill] (3.82415786385536, -1.3) circle (2pt);
\draw[fill] (3.82415786385536, -0.5) circle (2pt);
\draw[fill] (3.82415786385536, -1.7) circle (2pt);
\draw[fill] (3.82415786385536, -0.9) circle (2pt);
\draw[fill] (3.82415786385536, -2.1) circle (2pt);
\draw[fill] (5.53856943268329, -2.1) circle (2pt);
\draw[fill] (5.53856943268329, -0.5) circle (2pt);
\draw[fill] (5.53856943268329, -0.9) circle (2pt);
\draw[fill] (5.53856943268329, -1.7) circle (2pt);
\draw[fill] (5.53856943268329, -1.3) circle (2pt);
\draw[fill] (7.99700583759695, -1.7) circle (2pt);
\draw[fill] (7.99700583759695, -0.5) circle (2pt);
\draw[fill] (7.99700583759695, -1.3) circle (2pt);
\draw[fill] (7.99700583759695, -0.9) circle (2pt);
\draw[fill] (7.99700583759695, -2.1) circle (2pt);
\draw[fill] (9.76098943501711, -0.5) circle (2pt);
\draw[fill] (9.76098943501711, -0.9) circle (2pt);
\draw[fill] (9.76098943501711, -1.3) circle (2pt);
\draw[fill] (9.76098943501711, -2.1) circle (2pt);
\draw[fill] (9.76098943501711, -1.7) circle (2pt);
\draw[black] (0, -0.3) -- ++(0,-2);
\draw[black!60] (2.47207537293434, -0.3) -- ++(0,-2);
\draw[black!60] (3.82415786385536, -0.3) -- ++(0,-2);
\draw[black!60] (5.53856943268329, -0.3) -- ++(0,-2);
\draw[black!60] (7.99700583759695, -0.3) -- ++(0,-2);
\draw[black!60] (9.76098943501711, -0.3) -- ++(0,-2);
\draw[black] (10, -0.3) -- ++(0,-2);

\end{tikzpicture}
\caption[An example of grouping and rounding of large items for the $(1+\eps)$-competitive bin packing algorithm.]{An example of grouping and rounding of large items for $\eps = \sqrt{1/5}$. The top illustration shows the size of $24$ large items, denoted by the black dots, from $\sigma$ grouped into $5$ groups of $5$ items (except for the last group that contains $4$ items), according to a sorting of the items by size. The bottom illustration denotes the same items in the same grouping with their sizes rounded up as in $\sigma'$. Note that in the illustration the dots are placed at different heights to be able to clearly distinguish each point and has no other significance.}\label{fig:bpRound}
\end{figure}
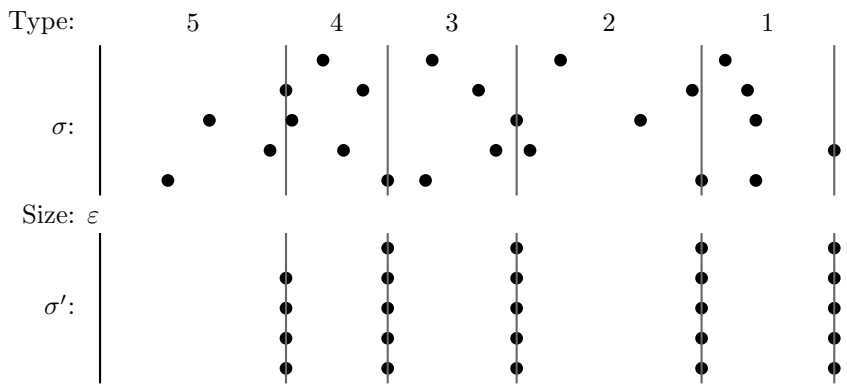

Let $\sigma'$ be the subsequence of $\sigma$ restricted to the large items with their sizes rounded up as per the scheme of Fernandez de la Vega and Lueker.
We now build a packing $\Sprime$. The type 1 items will be packed one item per bin. Let $B_1$ denote this set of bins. By definition, $|B_1| = \lceil\eps^2L\rceil$. Since large items have size at least $\eps$, $N \ge \eps L$. This implies the following fact.

\begin{fact}\label{fact:B1Size}
  $|B_1| \le \left\lceil\eps N\right\rceil$
\end{fact}

For the remaining \emph{large} items, i.e.\ types $2$ to $1/\eps^2$, in $\sigma'$, a packing, $B'_2$ that uses at most $N$ bins can be found efficiently~\cite{VegaL1981}. The packing of each bin $b_i \in B'_2$ can be described by a vector of length at most $1/\eps$, denoted $\mathbf{p_i}$, where each value in the vector ranges from $1$ to $1/\eps^2$ representing the type of each of the at most $1/\eps$ large items in $b_i$. This vector will be called a {\em bin pattern}. Let $B_2$ be a set of bins such that $|B_2| = |B'_2|$ and each $b_i \in B_2$ is assigned the bin pattern $b_i \in B'_2$. The items of $\sigma'$ can be assigned sequentially to the bins of $B_2$, using the following procedure. Initially, the bins of $B_2$ are all closed. For each $r_i \in \sigma'$, assign $r_i$ with type $t_i$ to the oldest open bin, $b_j$, such that there are less items of type $t_i$ packed in the bin than are described in $\mathbf{p_j}$. If no such bin exists, open a closed bin with a pattern that contains the type $t_i$ and pack $r_i$ in this bin. Note that such a bin must exist by the definition of $B_2$. 

The packing $\Sprime$ is defined to be $B_1 \cup B_2$ with the original (non-rounded up) sizes of the packed large items. The bins of $\Sprime$ are numbered from $1$ to $|\Sprime|$ based on the order that the bins would be opened when $\sigma'$ is processed sequentially. That is, for $i < j$ and every $b_i,b_j \in \Sprime$, there exists an $r_p \in b_i$ such that, for all $r_q \in b_j$, $p < q$. From Fact \ref{fact:B1Size} and $|B_2| \le N$, we have the following fact.

\begin{fact}\label{fact:SprimeSize}
  $|\Sprime| \le (1 + \eps)N + 1$
\end{fact}

We now extend $\Sprime$ to include the small items and define $\Salg$. Sequentially, by the order that the small items arrive, for each small item $r_i \in \sigma$, pack $r_i$ into $\Sprime$, using \textsc{next fit}. Additional bins are opened as necessary. The following lemma shows that $\Salg$ is a near-optimal packing. Note the this bound implies that $\Salg$ may pack one more bin than $(1 + 2\eps)$ times the optimal, making it an asymptotically $(1 + 2\eps)$-competitive packing. 

\begin{lemma}\label{lem:SalgSize}
  $|\Salg| \le (1 + 2\eps)N + 1$
\end{lemma}

\begin{proof}
  After packing the small items, if no new bins are opened then the claim follows from Fact \ref{fact:SprimeSize}. If there are additional bins opened, all the bins of $\Salg$, except possibly the last one, are filled to at least $(1-\eps)$. Since the total size of the items is at most $N$, we have $(|\Salg|-1)(1-\eps) \le N$ and, therefore, $|\Salg| \le \frac{N}{1-\eps} + 1 \le (1 + 2\eps)N + 1$. 
\end{proof}

We now define $\BPAA$.  It is defined given a $\sigma$, and an $\eps$, $0 < \eps \le 1/2$.  $\BPAA$ uses two (initially empty) sets of bins $L_1$ and $L_2$.  $L_1$ is the set of bins that pack small items and 0 or more large items. $L_2$ is the set of bins that pack only large items.  $\BPAA$ and the advice will be defined such that the items are packed exactly as $\Salg$.

With the first $N$ items, the advice bits indicate a bin pattern. These $N$ bin patterns will be the patterns of the bins in order from $\Salg$. As the bin patterns are received, they will be queued. Also, with each item, the advice bits indicate the type of the item. Small items will be of type $-1$. If the item is large, the bits of advice will also indicate if it is packed in $\Salg$ in a bin that also includes small items or not.  

During the run of $\BPAA$, bins will be opened and assigned bin patterns. The bins in each of the sets of bins are ordered according to the order in which they are opened. When a new bin is opened, it is assigned an empty bin pattern if the current item is small. If the current item is of type 1, the bin is assigned a type 1 bin pattern. Otherwise, the current item is of type $2$ to $1/\eps^2$, and the next pattern from the queue of bin patterns is assigned to the bin. Note that, by the definition of $\Salg$, this pattern must contain an entry for an item of the current type.

For each $r_i \in \sigma$, the items are packed and bins are opened as follows:

\paragraph{Small Items} 
For packing the small items, $\BPA$ maintains a pointer into the set $L_1$ indicating 
the bin into which it is currently packing small items. Additionally, the advice for the small items includes
a bit (the details of this bit will be explained subsequently) to indicate if this pointer should be moved to the next bin in $L_1$. If this is the case,
the pointer is moved prior to packing the small item and, if there is no next bin in $L_1$, a new bin with an empty pattern
is opened and added to $L_1$.  Then, the small item is packed into the bin referenced by the pointer. 

\paragraph{Large Items}  $\BPA$ receives an additional bit $y$ as advice that indicates if $r_i$ is packed in a bin in $\Salg$ that also includes small items.

{\it Type $1$ items: } 
If the item $r_i$ is packed into a bin with small items ($y = 1$), $r_i$ is packed in the oldest bin with an empty pattern. If no such bin exists, then $r_i$ is packed into a new bin that is added to $L_1$. If $r_i$ is packed into a bin without small items ($y = 0$), then $r_i$ is packed into a new bin that is added to $L_2$. In all the cases, the bin into which $r_i$ is packed is assigned a type 1 bin pattern. 

{\it Type $i > 1$ items:} Let  $t_i$ be the type of $r_i$. If $r_i$ is packed with small items ($y = 1$), then $r_i$ is packed into the oldest bin of $L_1$ such that the bin pattern specifies more items of type $t_i$ than are currently packed.  If no such bin exists, then $r_i$ is packed in the first bin with an empty bin pattern and the next bin pattern from the queue is assigned to this bin. If there are no empty bins, a new bin is added to pack $r_i$. If $r_i$ is not packed with small items ($y = 0$), $r_i$ is packed analogously but into the bins of $L_2$.   

The advice bit used to move the pointer for packing small items (see Section~\ref{bpAdvDef} for a formal definition) is defined so that $\BPA$ will schedule the same number of small items on each bin as $\Salg$. Further, $\BPA$ schedules both the small and large jobs in the order the arrive on the least recently opened bin just as $\Salg$ (see Figure \ref{fig:bpABPA}) which implies the following fact.

\begin{figure}[htb]
\centering
\begin{tikzpicture}[thick]

\draw (0,2) -- (0,0) -- (1,0) -- (1,2);
\draw (1.25,2) -- (1.25,0) -- (2.25,0) -- (2.25,2);
\draw (2.5,2) -- (2.5,0) -- (3.5,0) -- (3.5,2);
\draw[loosely dotted] (3.75,1) -- (4.75,1);
\draw (5,2) -- (5,0) -- (6,0) -- (6,2);
\draw[decorate,decoration={brace,mirror}] (0,-0.1) to node [midway,below] {$B_{2+}$} (6,-0.1); 

\draw (6.5,2) -- (6.5,0) -- (7.5,0) -- (7.5,2);
\draw (7.75,2) -- (7.75,0) -- (8.75,0) -- (8.75,2);
\draw (9,2) -- (9,0) -- (10,0) -- (10,2);
\draw[decorate,decoration={brace,mirror}] (6.5,-0.1) to node [midway,below] {$B_1$} (10,-0.1); 

\draw (10.5,2) -- (10.5,0) -- (11.5,0) -- (11.5,2);
\draw[decorate,decoration={brace,mirror}] (10.5,-0.1) to node [midway,below] {$B_{\text{small}}$} (11.5,-0.1); 

\draw[decorate,decoration={brace,mirror}] (0,-0.75) to node [midway,below] {$L_1 \cup L_2 \equiv \Salg$} (11.5,-0.75); 

\end{tikzpicture}
\caption[An illustration of the packing produced by $\BPAA$.]
{An illustration of the packing produced by $\BPAA$, $L_1 \cup L_2$, that is equivalent to the packing $\Salg$. $B_{2+}$ packs items of type $2$ to $1/\eps^2$ into $N$ bins. $B_1$ represents the set of $\eps N$ bins dedicated to packing type $1$ items and $B_{\text{small}}$ represents the (possibly empty) set of at most $\eps N + 1$ bins dedicated to packing the overflow of small items from the \textsc{next fit} packing of the small items into the bins of $B_1 \cup B_{2+}$.}\label{fig:bpABPA}
\end{figure}
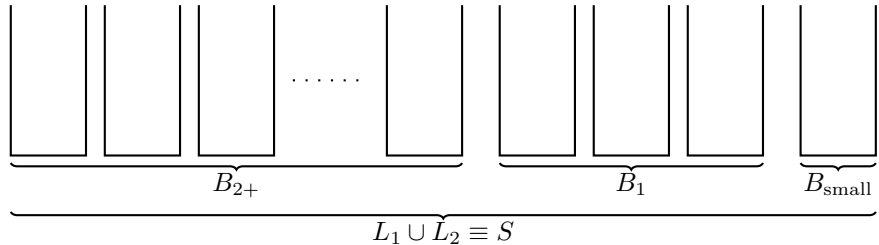

\begin{fact}\label{fact:L1L2isSOFF}
$L_1 \cup L_2$ is the same packing as $\Salg$. 
\end{fact}
Therefore, $|L_1 \cup L_2| \le (1 + 2\eps)N + 1$ by Lemma \ref{lem:SalgSize}.

\subsubsection{Formal Advice Definition.}\label{bpAdvDef}

\paragraph{Bin Patterns}
Instead of sending the entire vector representing a bin pattern, we enumerate all the possible vectors and the advice will be the index of the vector from the enumeration encoded in binary. The bin pattern vectors have a length of at most $1/\eps$ and there are at most $1/\eps^2$ different possible values.  
To ensure that the all vectors have the same length, a new value $\bpNull$ is used to pad vectors to a length of $1/\eps$. The increase the number of possible values per entry to $1/\eps^2 + 1$.  

The algorithm requires less than $\left\lceil\frac{1}{\eps}\log\left(\frac{1}{\eps^2} + 1\right)\right\rceil < \frac{1}{\eps}\log\left(\frac{2}{\eps^2}\right) + 1$
bits of advice per request to encode the index of the bin pattern from an enumeration of all possible bin patterns in binary.

\paragraph{Advice per Request}
In order to define the advice, for each bin $b_i \in \Salg$, we define a value $\kappa_i$ that is the number of small items packed in $b_i$.

Per request, the advice string will be $xyz$, where $x$ is $\left\lceil\log \left(1/\eps^2 + 1\right)\right\rceil < \log \left( 2/\eps^2 \right) + 1$ bits in length to indicate the type of the item; $y$ is $1$ bit in length to indicate whether the large items are packed with small items, or to indicate to small items whether or not to move to a new bin; $z$ has a length less than $\frac{1}{\eps}\log\left(\frac{2}{\eps^2}\right) + 1$ to indicate a bin pattern. $xyz$ is defined as follows for request $r_i$:

\sAdvice
{\bf $x$:}
&  & The type of $r_i$ encoded in binary.\\
{\bf $y$:}
& $r_i$ is a small item: & Let $s$ be the number of small jobs in $\left<r_1,\ldots,r_{i-1}\right>$.  If there exists an integer $1 \le j \le N$ such that $\sum_{k=1}^j\kappa_k = s$, then the first bit is a $1$.  Otherwise, the first bit is a $0$. \\
& $r_i$ is a large item: & 1, if $\kappa_i > 0$, where $b_i$ is the bin in which $r_i$ is packed in $\Salg$, i.e.\ $b_i$ packs small items. Otherwise, 0.\\
{\bf $z$:}
& $i \le N$ & The bits of $z$ encode a number in binary indicating the vector representing the bin pattern of the $i$-th bin opened by $\Sprime$. \\
& $i > N$ & Not used. All zeros. \\
\fAdvice

\subsection{Strict $(1+3\eps)$-Competitive Algorithm.}

$\BPA$ is defined such that it will behave in two different manners, depending on $N$ (i.e.\ the number of bins in an optimal packing) and $\eps$. One bit of advice per request, denoted by $w$, is used to distinguish between the two cases. The two cases are as follows.

\paragraph{Case 1: $N > 1/\eps$ ($w = 0$)} $\BPA$ will run $\BPAA$ as described previously. The only difference is that the advice per request for $\BPAA$ is prepended with an additional bit for $w$. Since $N > 1/\eps$, a single bin is at most $\eps N$ bins. Therefore, we get the following corollary to Lemma \ref{lem:SalgSize}.

\begin{corollary}\label{cor:SalgSize}
  $|\Salg| \le (1 + 3\eps)N$
\end{corollary}  

\paragraph{Case 2: $N \le 1/\eps$ ($w = 1$)} In this case, for each $r_i \in \sigma$, after $w$, the next $\left\lceil\log(1/\eps)\right\rceil$ bits of advice per request define the bin number in which $r_i$ is packed in an optimal packing. $\BPA$ will pack $r_i$ into the bin as specified by the advice. This case requires less than $\log(1/\eps) + 2 < \bpAdvUp$ (the upper bound on the amount of advice used per request in case 1) bits of advice per request and the packing produced is optimal.

The definition of the algorithm and the advice, Fact~\ref{fact:L1L2isSOFF} and Corollary~\ref{cor:SalgSize} prove Theorem~\ref{thm:bpaCR}. 

\section{Online Algorithms with Advice for Scheduling}\label{sec:sch}

In this section, we present a general framework for the online scheduling problem on $m$ identical machines. This framework depends on a positive $\eps < 1/2$, $U > 0$, and the existence of an optimal schedule $\Sstar$, where all jobs with a processing time greater than $U$ are scheduled on a machine without any other jobs. The framework will produce a schedule $\Salg$ such that, up to a permutation of the machines of $\Salg$, the load of machine $i$ in $\Salg$ is within $\eps L_i(\Sstar)$ of the load of machine $i$ in $\Sstar$, where $L_i(S)$ denotes the load of machine $i$ in the schedule $S$. This is done using $\schdBigOAdv$ bits of advice per request. We show that this nearly optimal schedule is sufficient for $(1+\eps)$-competitive algorithms for the makespan and minimizing the $\ell_p$ norm objectives, and a $\left(1/(1-\eps)\right)$-competitive algorithm for the machine cover objective. 

As is the case with the $(1+\eps)$-competitive algorithm for bin packing problem presented in Section~\ref{sec:bin}, the algorithms with a competitive ratio of $(1+\eps)$ (resp. $(1/(1-\eps))$) presented in this section could use advice that was based on a schedule produced by a PTAS for the desired objective function as opposed to an optimal schedule.

If $\log m$ is less than the number of bits of advice per request to be given to the algorithm, then the trivial algorithm with advice that encodes, for each job, the machine number to schedule that job could be used to obtain a $1$-competitive algorithm instead of the framework presented in this section.

\subsection{General Framework}
The machines are numbered from $1$ to $m$. Given an $\eps$, $0< \eps < 1/2$, and $U > 0$, the requested jobs will be classified into a constant number of {\em types}, using a geometric classification. $U$ is a bound which depends on the objective function of the schedule.  Formally, a job is of type $i$ if its processing time is in the interval $(\eps(1+\eps)^i U,\eps(1+\eps)^{i+1}U]$ for $i\in[0,\lceil\log_{1+\eps}\frac1\eps\rceil)$. These jobs will be called {\em large jobs} (see Figure \ref{fig:schedGroup}). Jobs with processing times at most $\eps U$ will be considered {\em small jobs} and have a type of $-1$.  Jobs with processing times greater than $U$ will be considered {\em huge jobs} and have a type of $\msTypes$. The online algorithm does not need to know the actual value of the threshold $U$.

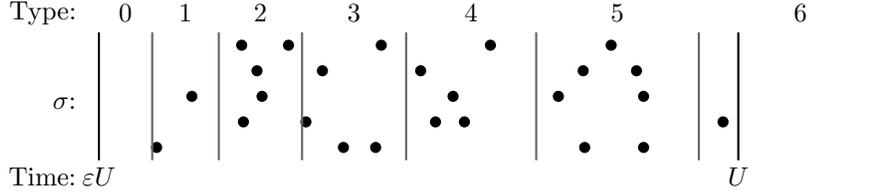
\begin{figure}[htb]
\centering
\begin{tikzpicture}[thick,scale=0.85]
\draw[black] (0, 0.2) -- ++(0,2);
\draw[fill] (0.902132177725434, 0.4) circle (2pt);
\draw[fill] (1.45366716664284, 1.2) circle (2pt);
\draw[fill] (2.23257064819336, 2) circle (2pt);
\draw[fill] (2.2599522722885, 0.8) circle (2pt);
\draw[fill] (2.47207537293434, 1.6) circle (2pt);
\draw[fill] (2.5510832760483, 1.2) circle (2pt);
\draw[fill] (2.96437476761639, 2) circle (2pt);
\draw[fill] (3.23697155807167, 0.8) circle (2pt);
\draw[fill] (3.49513975903392, 1.6) circle (2pt);
\draw[fill] (3.82415786385536, 0.4) circle (2pt);
\draw[fill] (4.32760030496866, 0.4) circle (2pt);
\draw[fill] (4.41621185280383, 2) circle (2pt);
\draw[fill] (5.0322353374213, 1.6) circle (2pt);
\draw[fill] (5.26406505145133, 0.8) circle (2pt);
\draw[fill] (5.53856943268329, 1.2) circle (2pt);
\draw[fill] (5.7159958826378, 0.8) circle (2pt);
\draw[fill] (6.12248234916478, 2) circle (2pt);
\draw[fill] (7.18547832686454, 1.2) circle (2pt);
\draw[fill] (7.57348593790084, 1.6) circle (2pt);
\draw[fill] (7.59700583759695, 0.4) circle (2pt);
\draw[fill] (8.01166283432394, 2) circle (2pt);
\draw[fill] (8.4076991211623, 1.6) circle (2pt);
\draw[fill] (8.51852759253234, 1.2) circle (2pt);
\draw[fill] (8.51912145968527, 0.4) circle (2pt);
\draw[fill] (9.76098943501711, 0.8) circle (2pt);
\draw[black] (10, 0.2) -- ++(0,2);

\draw[black!60] (0.8333333333, 0.2) -- ++(0,2);
\draw[black!60] (1.875000000, 0.2) -- ++(0,2);
\draw[black!60] (3.177083333, 0.2) -- ++(0,2);
\draw[black!60] (4.804687500, 0.2) -- ++(0,2);
\draw[black!60] (6.839192708, 0.2) -- ++(0,2);
\draw[black!60] (9.382324219, 0.2) -- ++(0,2);
\draw[black!60] (12.56123861, 0.2) -- ++(0,2);

\node[left] at (-0.2, 2.5) {Type:};
\node at (0.4166666667, 2.5) {$0$};
\node at (1.354166667, 2.5) {$1$};
\node at (2.526041667, 2.5) {$2$};
\node at (3.990885417, 2.5) {$3$};
\node at (5.821940104, 2.5) {$4$};
\node at (8.110758464, 2.5) {$5$};
\node at (10.97178141, 2.5) {$6$};
\node[left] at (-0.2,1.1) {$\sigma$:};
\node[left] at (-0.2,-0.05) {Time:};
\node at (0,-0.05) {$\eps U$};
\node at (10,-0.05) {$U$};
\end{tikzpicture}
\caption[An example of grouping of the large jobs for the $(1+\eps)$-competitive scheduling framework.]{An example of the geometric grouping of large jobs, jobs with a processing time in the range $(\eps U, U]$, for $\eps = 1/4$. The illustration shows the processing time of the $25$ large jobs, denoted by the black dots, from $\sigma$ grouped into $7$ groups, according to a sorting of the jobs by processing time. The $i$-th group consists of jobs with a processing time in the interval $(\eps(1+\eps)^i U,\eps(1+\eps)^{i+1}U]$. 
Even though the range for type $6$ is greater than $U$, only jobs with a processing time at most $U$ will be assigned type $6$.}\label{fig:schedGroup}
\end{figure}

Let $\Sstar$ be an optimal schedule for the input at hand. In what follows, we will define a schedule $\Sprime$ from $\Sstar$ such that, for all $i$, $L_i(\Sprime) \in [L_i(\Sstar) - \eps U, L_i(\Sstar) + \eps U]$. Then, based on $\Sprime$, we will define a schedule $\Salg$ such that $L_i(\Salg) \in [(1-\eps)L_i(\Sstar) - \eps U, (1+\eps)L_i(\Sstar) + \eps U]$. The advice will be defined so that the online algorithm will produce the schedule $\Salg$.

\medskip 
 
The framework makes the following assumption.  For each of the objective functions that we consider, we will show that
there always exists an optimal schedule for which this assumption holds. 
\begin{assumption} \label{ass:hugeJob}
$\Sstar$, the optimal schedule on which the framework is based, is a schedule such that each huge job is scheduled on a machine that does not contain any other job.
\end{assumption}  

The general framework is defined given a $\sigma$, an $\eps$, a $U$, and an $\Sstar$ under Assumption~\ref{ass:hugeJob}. For the schedule $\Sstar$, we assume without loss of generality that machines are numbered from $1$ to $m$, according to the following order. Assign to each machine the request index of the first large job scheduled on it. Order the machines by increasing order of this number. Machines on which no large job is scheduled are placed at the end in an arbitrary order. 

We define $\Sprime$ by removing the small jobs from $\Sstar$.
$\Sprime$ can be described by $m$ patterns, one for each machine.  Each such pattern will be called a {\em machine pattern}.  For machine $i$, $1 \le i \le m$, the machine pattern 
indicates that (1) the machine schedules large or huge jobs, or (2) an empty machine (such a machine may schedule only small jobs).
In the first case, the machine pattern is a vector with one entry per large or huge job scheduled on machine $i$ in $\Sprime$.  These entries will be the job types of these jobs on machine $i$ ordered from smallest to largest.  Let $v$ denote the maximum length of the machine pattern vectors for $\Sprime$.  The value of $v$ will be dependent on the objective function and $U$.  
We later show that for all the objective functions we consider, $v \leq1/\eps + 1$.

We now extend $\Sprime$ to also include the small jobs. Figure~\ref{fig:smallJobs} depicts the proof of the following lemma.

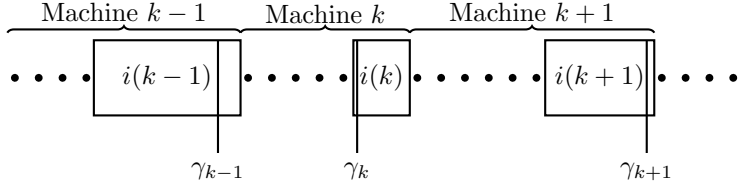
\begin{figure}[htb]
\centering
\begin{tikzpicture}[thick, dot/.style={line cap=round,dash pattern=on 0pt off 3\pgflinewidth, line width=3pt}]

  \draw[decorate,decoration={brace}] (-0.1,1.6) to node[midway,above] {Machine $k-1$} (3,1.6);
  \draw[dot] (0,1) -- (1,1);
  \draw (1.05,0.5) rectangle (3,1.5) node[midway] {$i(k-1)$};
  \draw (2.7, 0) node[below] {$\gamma_{k-1}$} to (2.7, 1.5);  \draw[decorate,decoration={brace}] (3,1.6) to node[midway,above] {Machine $k$} (5.25,1.6);
  \draw[dot] (3.1,1) -- (4.4,1); 
  \draw (4.5,0.5) rectangle (5.25,1.5) node[midway] {$i(k)$};
  \draw (4.55, 0) node[below] {$\gamma_{k}$} to (4.55, 1.5);
  \draw[decorate,decoration={brace}] (5.25,1.6) to node[midway,above] {Machine $k+1$} (8.5,1.6);
  \draw[dot] (5.35,1) -- (7,1); 
  \draw (7.05,0.5) rectangle (8.5,1.5) node[midway] {$i(k+1)$};
  \draw (8.4, 0) node[below] {$\gamma_{k+1}$} to (8.4, 1.5);
  \draw[dot] (8.6,1) -- (9.6,1); 

\end{tikzpicture}
\caption[An illustration of the proof Lemma~\ref{lem:all}.]{An illustration of the proof of Lemma~\ref{lem:all}. In the illustration, the rectangles represent the processing time of the small jobs as they are ordered in $\sigma$. The value $\gamma_i$ indicates the total processing time of the small jobs scheduled on the first $i$ machines in an optimal schedule. The difference between $\gamma_i$ and $\gamma_{i-1}$ ($\gamma_0 = 0$) being the processing time of the smalls jobs on machine $i$ which is denoted by $y_i$ in Lemma~\ref{lem:all}. The small jobs are assigned in a next fit manner to the machines of $\Sprime$ such that the small jobs scheduled on machine $k$ include all the small jobs from the first small job immediately after the last small job scheduled on machine $k-1$ ($i(k-1)$) to the small job ($i(k)$) such that the total processing time of all the small jobs prior to and including $i(k)$ is at least $\gamma_k$. This ensures that the total processing time of the small jobs assigned to machine $k$ is within $\eps U$ of the processing time of the small jobs on machine $k$ in the optimal schedule.}  \label{fig:smallJobs}
\end{figure}

\begin{lemma}
\label{lem:all}
The small jobs of $\sigma$ can be scheduled on the machines of $\Sprime$ sequentially in a next fit manner from machine $1$ to machine $m$, such that the load for each machine $i$ will be in $[L_i(\Sstar) - \eps U, L_i(\Sstar) + \eps U]$.
\end{lemma}
\begin{proof}
Consider the small jobs in the order in which they arrive. 
Denote the processing time of the $j$th small job in this order by $x_j$ for $j=1,\dots$.
For $i=1,\dots,m$, let $y_i$ be the total processing time of small jobs assigned to machine $i$ in $\Sstar$.
Let $i(0)=0$, and
for $k=1,\dots,m$, let $i(k)$ be the minimum index such that $\sum_{j=1}^{i(k)} x_j \ge \sum_{i=1}^k y_i$.
Finally, for $k=1,\dots,m$, assign the small jobs $i(k-1)+1,\dots,i(k)$ to machine $k$.
(If $i(k)=i(k-1)$, machine $k$ receives no small jobs.)

By the definition of $i(k)$ and the fact that all small jobs have a processing time at most $\eps U$, 
the total processing time of small jobs assigned to machines $1,\dots,k$ is in $[\sum_{i=1}^k y_i, \sum_{i=1}^k y_i+\eps U]$
for $k=1,\dots,m$.
By taking the difference between the total assigned processing time for the first $k-1$ and for the first $k$ machines,
it immediately follows that the total processing time of small jobs assigned to machine $k$ is
in $[y_k-\eps U, y_k+\eps U]$.
\end{proof}
Note that some machines may not receive any small jobs in this process. We will use the advice bits to separate the machines that receive small jobs from the ones that do not, so that we can assign the small jobs to consecutive machines.

We now define the schedule $\Salg$, using the following procedure. Assign the machine patterns of $\Sprime$ in the same order to the machines of $\Salg$. For each large or huge job $r_i \in \sigma$, in the order they appear in $\sigma$, assign $r_i$ with type $t_i$ to the first machine in $\Salg$ such that the number of jobs with type $t_i$ currently scheduled is less than the number of jobs of type $t_i$ indicated by the machine pattern. After all the large and huge jobs have been processed, assign the small jobs to the machines of $\Salg$ exactly as they are assigned in $\Sprime$ in Lemma \ref{lem:all}. 

\begin{lemma}\label{lem:schedS}
  For $1 \le i \le m$, $L_i(\Salg)\in [(1-\eps)L_i(\Sstar) - \eps U, (1+\eps)L_i(\Sstar) + \eps U]$.
\end{lemma}
\begin{proof}
 By Lemma \ref{lem:all} and the fact that 
 jobs of the same type differ by a factor of at most $1+\eps$, 
 we have
$L_i(\Salg)\in \left[\frac{1}{1+\eps}L_i(\Sstar) - \eps U, (1+\eps)L_i(\Sstar) + \eps U\right]$. The claim follows since $1/(1+\eps) > 1-\eps$ for $\eps>0$.
\end{proof}

We have thus shown that in $\Salg$ the load on 
every 
machine is very close to the optimal load (for an appropriate choice of $U$). Note that this statement
is 
independent of the objective function. This means if we can find such a schedule $S$ online
with a good value of $U$, we can achieve 
our goal 
for every function of the form $\sum_{i=1}^m f(L_i)$, where $f$ satisfies the property that
if $x\le (1+\eps)y$ then $f(x)\le (1+O(1)\eps)f(y)$. 

We now define the online algorithm with advice for the general framework, which produces a schedule equivalent to $\Salg$ up to a permutation of the machines. For simplicity of presentation, we assume that this permutation is the identity permutation.

\begin{figure}[htb]
\centering
\begin{tikzpicture}[thick]
  \node at (-0.5,2) {$\Salg$:};
  \node[below,align=center] at (-0.5,0) {$i,\kappa_i:$};
  \draw[ultra thick] (0,2) -- (0,0) to node[midway,below,align=center,black] (1z) {$1,0$} (1,0) -- (1,2);
  \draw[ultra thick,double] (1.25,2) -- (1.25,0) to node[midway,below,align=center,black] (2z) {$2,4$} (2.25,0) -- (2.25,2);
  \draw[ultra thick,double] (2.5,2) -- (2.5,0) to node[midway,below,align=center,black] (3z) {$3,1$} (3.5,0) -- (3.5,2);
  \draw[ultra thick] (3.75,2) -- (3.75,0) to node[midway,below,align=center,black] (4z) {$4,0$} (4.75,0) -- (4.75,2);
  \draw[line cap=round,dash pattern=on 0pt off 2\pgflinewidth, line width=4pt] (5,1) -- (6,1);
  \node at (5.25, -0.25) (5z) {};
  \node at (5.5, -0.25) (6z) {}; 
  \draw[ultra thick] (6.25,2) -- (6.25,0) to node[midway,below,align=center,black] (m-1z) {$m-1,0$} (7.25,0) -- (7.25,2);
  \draw[ultra thick,double] (7.5,2) -- (7.5,0) to node[midway,below,align=center,black] (mz) {$m,2$} (8.5,0) -- (8.5,2);

  \node at (-0.5,-1.75) {$Z$:};
  \draw[ultra thick] (0,-1.75) -- (0,-3.75) to node[midway,below,align=center] {$1$} (1,-3.75) -- (1,-1.75);
  \draw[ultra thick] (0,-1.75) -- (0,-3.75) to node[midway,below,align=center,black] {$1$} (1,-3.75) -- (1,-1.75);
  \node at (0.5, -1.75) (1s) {};
  \draw[ultra thick] (1.25,-1.75) -- (1.25,-3.75) to node[midway,below,align=center] {$2$} (2.25,-3.75) -- (2.25,-1.75);
  \draw[ultra thick] (1.25,-1.75) -- (1.25,-3.75) to node[midway,below,align=center,black] {$2$} (2.25,-3.75) -- (2.25,-1.75);
  \node at (1.75, -1.75) (2s) {};
  \draw[ultra thick] (2.5,-1.75) -- (2.5,-3.75) to node[midway,below,align=center] {$3$} (3.5,-3.75) -- (3.5,-1.75);
  \draw[ultra thick] (2.5,-1.75) -- (2.5,-3.75) to node[midway,below,align=center,black] {$3$} (3.5,-3.75) -- (3.5,-1.75);
  \node at (3, -1.75) (3s) {};
  \draw[ultra thick] (3.75,-1.75) -- (3.75,-3.75) to node[midway,below,align=center] {$4$} (4.75,-3.75) -- (4.75,-1.75);
  \draw[ultra thick] (3.75,-1.75) -- (3.75,-3.75) to node[midway,below,align=center,black] {$4$} (4.75,-3.75) -- (4.75,-1.75);
  \node at (4.25, -1.75) (4s) {};
  \draw[line cap=round,dash pattern=on 0pt off 2\pgflinewidth, line width=4pt] (5,-2.75) -- (6,-2.75); 
  \node at (5.5, -1.75) (5s) {};
  \node at (5.75, -1.75) (6s) {};
  \draw[ultra thick,double] (6.25,-1.75) -- (6.25,-3.75) to node[midway,below,align=center,black] {$m-1$} (7.25,-3.75) -- (7.25,-1.75);
  \node at (6.75, -1.75) (m-1s) {};
  \draw[ultra thick,double] (7.5,-1.75) -- (7.5,-3.75) to node[midway,below,align=center,black] {$m$} (8.5,-3.75) -- (8.5,-1.75);
  \node at (8, -1.75) (ms) {};
  \node[below,align=center] at (-0.5,-3.75) {$i:$};
  \draw[thick,->] (1z) -- (1s);
  \draw[thick,->] (2z) -- (ms);
  \draw[thick,->] (3z) -- (m-1s);
  \draw[thick,->] (4z) -- (2s);
  \draw[thick,->] (5z) -- (3s);
  \draw[thick,->] (6z) -- (4s);
  \draw[thick,->] (m-1z) -- (5s);
  \draw[thick,->] (mz) -- (6s);
  \node[draw,fill=black!30,anchor=west,single arrow,minimum height=5.5cm] at (0,-2.75) {$\kappa_i = 0$};
  \node[draw,fill=black!30,anchor=west,single arrow,shape border rotate=180,minimum height=3cm] at (5.5,-2.75) {$\kappa_i > 0$};
\end{tikzpicture}
\caption[An illustration of the permutation of the machines of $\Salg$ when scheduled by the online algorithm.]{In order to produce the exact same schedule as $\Salg$, the online algorithm must permute the machines of $\Salg$ based on the number of small jobs scheduled per machine (denoted by $\kappa_i$ for machine $i$). This figure illustrates such a permutation. $Z$ denotes the schedule produced by the online algorithm. Note that, in $Z$, machines with no small jobs, single line ($\kappa_i = 0$), are in the same relative order as $\Salg$ and machines with small jobs, double line ($\kappa_i > 0$), are in reverse relative order. This allows the jobs for machines without small jobs to be scheduled from left to right and jobs for machines with small jobs to be scheduled from right to left.}\label{fig:schedPerm}
\end{figure}
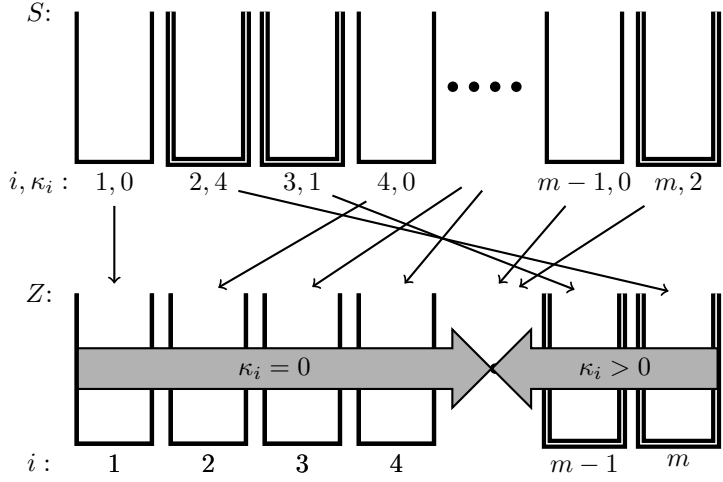

For the first $m$ requests, the general framework receives as advice a machine pattern and a bit $y$, which indicates whether this machine contains small jobs or not (see Figure~\ref{fig:schedPerm}). 
For $r_j$, $1 \le j \le m$, if 
$y=0$, the framework assigns the machine pattern to the {\em highest} machine number without an assigned pattern.  Otherwise, 
the framework will assign the machine pattern to the {\em lowest} machine number without an assigned pattern.  For each request $r_i$ in $\sigma$, the type of $r_i$, denoted by $t_i$, is received as advice.  The framework schedules $r_i$, according to $t_i$, as follows: 

\paragraph{Small Jobs $(t_i = -1)$} For scheduling the small jobs, the algorithm maintains a pointer to a machine (initially machine $m$) indicating the machine that is currently scheduling small jobs. With each small job, the algorithm gets a bit of advice $x$ that indicates if this pointer should be moved to the machine with the preceding serial number. If so, the pointer is moved prior to scheduling the small job.  Then, $r_i$ is scheduled on the machine referenced by the pointer.

\paragraph{Large and Huge Jobs $(0 \le t_i \le \lceil\log_{1+\eps}\frac1\eps\rceil)$} 
The algorithm schedules $r_i$ on a machine where the number of jobs of type $t_i$ is less than the number indicated by its pattern.  

\subsubsection{Formal advice definition.}\label{schAdvDef}

\paragraph{Machine Patterns}
For the first $m$ requests, a machine pattern is received as advice.  Specifically, all possible machine patterns will be enumerated and the id of the pattern, encoded in binary, will be sent as advice for each machine.  For large jobs, there are at most $v$ jobs in a machine pattern vector, and each job has one of $\msTypes$ possible types.  The machine patterns can be described with the jobs ordered from smallest to largest since the order of the jobs on the machine is not important.  This is equivalent to pulling $v$ names out of  $\msTypes + 1$ names (one name to denote an empty entry and $\msTypes$ names for each of the large job types), where repetitions are allowed and order is not significant.  Therefore, there are 
\[
\binom{v + \msTypes}{\msTypes} \le \msTypes^v 
\]
different possible machine patterns for the machines scheduling large jobs.  Additionally, there is a machine pattern for machines with only small jobs and a machine pattern for machines with only a huge job.  Hence, at most 
$\beta(v) \le \left\lceil \log(2 + \msTypes^v) \right\rceil < v\log\left(\frac{3\log(1/\eps)}{\log(1+\eps)}\right) + 1$ bits are required to encode the index of a machine pattern in an enumeration of all possible machine patterns in binary. As we show in the following, for the cases of makespan, machine cover and $\ell_p$ norm, $v \le 1/\eps + 1$ and $\beta(v) < \frac{1 + \eps}{\eps}\log\left(\frac{3\log(1/\eps)}{\log(1+\eps)}\right) + 1$.

\paragraph{Advice per Request}
In order to define the advice, for each machine $m_i \in \Salg$, we define a value $\kappa_i$ that is the number of small jobs scheduled on $m_i$.

Per request, the advice string will be $wxyz$, where $w$ has a length of \\
$\left\lceil \log (2+\msTypes) \right\rceil < \log\left(\frac{3\log(1/\eps)}{\log(1+\eps)}\right) + 1$ bits to indicate the job type, $x$ and $y$ 
are 1 bit in length (as described above), and $z$ has a length of $\beta(v)$ bits to indicate a machine pattern. $wxyz$ is defined as follows for request $r_i$:

\sAdvice
{\bf $w$:}
& & A number in binary representing the type of $r_i$.\\
{\bf $x$:}
& $r_i$ is a small job: & $x = 1$ if the small job should be scheduled on the next machine.  Otherwise, $x = 0$.  More formally, let $s$ be the number of small jobs in $\left<r_1,\ldots,r_{i-1}\right>$.  If there exists and an integer $1 \le j \le m$ such that $\sum_{k=1}^j\kappa_k = s$, then $x = 1$.  Otherwise, $x = 0$. \\
& otherwise: & $x$ is unused and the bit is set to $0$. \\
{\bf $y$:}
& $i \le m$: & If $\kappa_i > 0$, $y = 0$.  Otherwise, $y = 1$. \\
& $i > m$: & This bit is unused and set to 0. \\
{\bf $z$:}
& $i \le m$: & $z$ is a number in binary indicating the machine pattern of machine $i$ in $\Sprime$.\\
& $i > m$: & $z$ is unused and all the bits are set to $0$. \\
\fAdvice

\begin{fact}\label{fact:schGenAdv}
This framework uses less than 
$\log\left(\frac{3\log(1/\eps)}{\log(1+\eps)}\right) + \beta(v) + 3$ bits of advice per request.  
\end{fact}

The following theorem, which follows immediately from the definition of the general framework and Lemma~\ref{lem:schedS}, summarizes the main result of this section.
\begin{theorem}\label{thm:genFrame}
For any  $\sigma$, an $\eps$, $0 < \eps < 1/2$, and a $U > 0$ such that there exists an $\Sstar$ under  Assumption~\ref{ass:hugeJob},  the general framework schedules $\sigma$ such that 
for all machines, $1 \le i \le m$, $L_i(\Salg)\in [(1-\eps)L_i(\Sstar) - \eps U, (1+\eps)L_i(\Sstar) + \eps U]$.
 \end{theorem}

\subsection{Minimum Makespan}\label{sec:DetMM}
For minimizing the makespan on $m$ identical machines, we will define $U = \OPT$, where $\OPT$ is the minimum makespan for $\sigma$.

\begin{fact}\label{fact:HugeMS}
If $U = \OPT$, there are no huge jobs as the makespan is at least as large as the largest processing time of all the jobs.
\end{fact}
By the above fact, we know that Assumption~\ref{ass:hugeJob} holds.

\begin{lemma}\label{cl:MSnumJobs}
The length of the machine pattern vector is at most $\frac1\eps$.
\end{lemma}

\begin{proof}
This lemma follows from the fact that all large jobs have a processing time greater than $\eps U = \eps\OPT$ and that a machine in $\Sstar$ with more than $\frac1\eps$ jobs with processing times greater than $\eps\OPT$ is more than the maximum makespan, a contradiction.
\end{proof}

From Lemma \ref{cl:MSnumJobs}, $v = \frac1\eps$.  Using this value with Fact \ref{fact:schGenAdv} of the general framework, gives the following.
\begin{fact}
The online algorithm with advice, based on the general framework, uses at most
 $\frac{2}{\eps}\left(\log\left(\frac{3\log(1/\eps)}{\log(1+\eps)}\right)\right) + 4$
bits of advice per request.
\end{fact}

\begin{theorem}\label{thm:ms}
Given a request sequence $\sigma$, $U = \OPT$ and an $\eps$, $0 < \eps < 1/2$, the online algorithm with advice, based on the general framework schedules the jobs of $\sigma$ such that the online schedule has a makespan of at most $(1 + 2\eps)\OPT$.
\end{theorem}

\begin{proof}
By Fact~\ref{fact:HugeMS}, Assumption~\ref{ass:hugeJob} holds and Theorem~\ref{thm:genFrame} applies.  

Let $j$ be a machine with the maximum load in $\Sstar$.  By Theorem~\ref{thm:genFrame}, $L_i (\Salg) \le (1 + \eps)L_i(\Sstar) + \eps U \le (1 + 2\eps)\OPT$ as $U = \OPT = L_j(\Sstar) \ge L_i(\Sstar)$ for all $1 \le i \le m$.  
\end{proof}

\subsection{Machine Covering}\label{sec:DetMC}

For maximizing the minimum load, i.e.\ machine covering, on $m$ identical machines, we will define $U = \OPT$, where $\OPT$ is the load of the machine with the minimum load in $\Sstar$.

\begin{lemma}\label{cl:MConeHugeJob}
There exists an optimal schedule $S$ such that any job with processing time more than that of the minimum load, i.e.\ a huge job, will be scheduled on a machine without any other jobs.
\end{lemma}

\begin{proof}
In $S$, let machine $i$ be the machine with the minimum load. Note that by definition a huge job has a processing time that is more than the minimum load. Therefore, machine $i$ cannot contain a huge job. Assume that scheduled on some machine $j \ne i$ is a huge job and one or more large or small jobs. 
We will denote the set of non-huge jobs scheduled on machine $j$ by $J$.  We will define another schedule $S^*$ to be the same schedule as $S$ for all the machines but $i$ and $j$.  In $S^*$, machine $i$ will schedule the same jobs as in $S$ plus all the jobs in $J$ and machine $j$ will only schedule the huge job scheduled on machine $j$ in $S$.  
The load on machine $j$ in $S^*$ is greater than $\OPT$ as it contains a huge job and the load on machine $i$ in $S^*$ is greater than $\OPT$ given that it was $\OPT$ in $S$ and jobs were added to it in $S^*$.  If the load of machine $i$ in $S$ is a unique minimum, then $S^*$ contradicts the optimality of $S$.  Otherwise, there exists another machine, $k \ne i$ and $k \ne j$, with the same load as $i$ in $S$.  Machine $k$ has the same load in $S^*$ as it does in $S$.  Therefore, $S^*$ is an optimal schedule.  This process can be repeated until a contradiction is found or an optimal schedule is created such that no huge job is scheduled on a machine with any other jobs.
\end{proof}

\begin{lemma}\label{cl:MCnumJobs}
There exists an optimal schedule $S$ such that there are at most $1 + \frac1\eps$ non-small jobs scheduled on each machine and huge jobs are scheduled on a machine without any other jobs.
\end{lemma}

\begin{proof}
By Lemma~\ref{cl:MConeHugeJob}, we can transform any optimal schedule $S$ to an optimal schedule $S'$, where all the huge jobs are scheduled on machines without any other jobs.

In $S'$, let machine $i$ be the machine with the minimum load and assume that some machine $j \ne i$ has more than $1 + \frac1\eps$ large jobs. We will define another schedule $S^*$ to be the same schedule as $S'$ for all the machines but $i$ and $j$. Note that machine $i$ has at most $\frac1\eps$ large jobs scheduled and, since its load is $U$, it cannot contain a huge job because huge jobs have processing times more than $U$. In $S^*$, machine $i$ will schedule the same jobs as $S'$ plus all the small jobs and the largest job scheduled on machine $j$ in $S'$.  The load on machine $j$ in $S^*$ is greater than $\OPT$ as it still has at least $1 + \frac1\eps$ large jobs scheduled on it and the load on machine $i$ in $S^*$ is greater than $\OPT$ given that it was $\OPT$ in $S'$ and jobs were added to it in $S^*$.  If the load of machine $i$ in $S'$ is a unique minimum, then $S^*$ contradicts the optimality of $S'$.  Otherwise, there exists another machine, $k \ne i$ and $k \ne j$, with the same load as $i$ in $S'$.  Machine $k$ has the same load in $S^*$ as it does in $S'$.  Therefore, $S^*$ is an optimal schedule.  This process can be repeated until a contradiction is found or an optimal schedule is created such that no machine has more than $1 + \frac1\eps$ non-small jobs scheduled.
\end{proof}

From Lemma \ref{cl:MCnumJobs}, $v = 1 + \frac1\eps$.  Using this value with Fact \ref{fact:schGenAdv} of the general framework gives the following.
\begin{fact}
The online algorithm with advice, based on the general framework, uses at most
 $\frac{3}{\eps}\left(\log\left(\frac{3\log(1/\eps)}{\log(1+\eps)}\right)\right) + 4$
bits of advice per request.
\end{fact}

\begin{theorem}
Given a request sequence $\sigma$, $U = \OPT$ and an $\eps$, $0 < \eps < 1$, the online algorithm with advice, based on the general framework, schedules the jobs of $\sigma$ such that the online schedule has a machine cover at least $(1 - 2\eps)\OPT$.
\end{theorem}

\begin{proof}
By Lemma~\ref{cl:MCnumJobs}, 
Assumption~\ref{ass:hugeJob} holds and Theorem~\ref{thm:genFrame} applies.  

Let $j$ be a machine with the minimum load in $\Sstar$.  By Theorem~\ref{thm:genFrame}, $L_i (\Salg) > (1 - \eps)L_i(\Sstar) - \eps U \ge (1 - 2\eps)\OPT$ as $\OPT = L_j(\Sstar) \le L_i(\Sstar)$ for all $1 \le i \le m$.  
\end{proof}

\subsection{The $\ell_p$ Norm}\label{sec:DetLP}

For minimizing the $\ell_p$ norm on $m$ identical machines, we will define $U = \frac{W}{m}$, where $W$ is the total processing time of all the jobs.

For completeness, we first prove the following technical lemma about convex functions. 
\begin{lemma}
\label{lem:convex}
Let $f$ be a convex function, and let $x_0 > y_0 \ge 0$.
Let $v < x_0-y_0$. Then $f(x_0-v) + f(y_0+v) < f(x_0) + f(y_0)$.
\end{lemma}
\begin{proof}
We need to show that $f(x_0)-f(x_0-v) > f(y_0+v)-f(y_0)$ for $y_0<x_0$ and $0<v<x_0-y_0$.

Suppose first that $v<(x_0-y_0)/2$.
Due to the mean value theorem, there exist values $\theta_1\in[y_0,y_0+v],\theta_2\in[x_0-v,x_0],$ such that
$f'(\theta_1)=(f(y_0+v)-f(y_0))/v$ and 
$f'(\theta_2)=(f(x_0)-f(x_0-v))/v$.
Since $f$ is convex (so $f''(x) \ge 0$) and $\theta_1\le y_0+v < x_0-v\le \theta_2$, we have $f'(\theta_1)< f'(\theta_2)$, proving the claim.

If on the other hand $v\ge(x_0-y_0)/2$, then define 
$w=x_0-(y_0+v) < (x_0-y_0)/2$ and note that $x_0-w=y_0+v$ and $y_0+w=x_0-v$.
The claim $f(x_0)-f(y_0+v) > f(x_0-v)-f(y_0)$ can now be shown exactly as above.
\end{proof}

\begin{lemma}\label{cl:LPadjSched}
For any schedule $S$,
moving a job from a machine where it is assigned together with a set of jobs of total size at least $W/m$ 
to a machine with minimum load strictly improves the $\ell_p$ norm.
\end{lemma}

\begin{proof}
Denote the schedule after the move by $S'$.
We show that $\sum_{i=1}^m f(L_i(S')) < \sum_{i=1}^m f(L_i(S))$, where $f(x)=x^p$ (for some $p>1$).
Denote the size of the job to be moved by $v>0$, the current load of its machine by $x_0$, where $x_0-v\ge W/m$
by assumption, and the current minimum load by $y_0<W/m$. Now we can apply Lemma \ref{lem:convex}.
\end{proof}

The following corollary follows from Lemma~\ref{cl:LPadjSched}. 

\begin{corollary}\label{cor:avgIsOpt}
For any schedule $S$, $||S||_p \ge (\sum_{i=1}^m (W/m)^p)^{1/p}$.
\end{corollary}
\begin{proof}
We apply Lemma \ref{cl:LPadjSched} repeatedly (if possible, i.e.\ if the load is not already exactly $W/m$ on every machine) and also allow parts of jobs to be moved (everything that is above a load of $W/m$ on some machine). Eventually we reach a flat schedule with a load of $W/m$ everywhere, and the $\ell_p$ norm is improved in every step.
\end{proof}

\begin{lemma}\label{cl:LPoneHugeJob}
In any optimal schedule $S$, any job with processing time greater than $\frac{W}{m}$, i.e.\ a huge job, will be scheduled on a machine without any other jobs.
\end{lemma}

\begin{proof}
There can  be at most $m-1$ huge jobs, else the total processing time of the jobs would be more than $W$. In a schedule with a huge job, the machine with the minimum load must have a load  less than $\frac{W}{m}$ (and cannot contain a huge job), else the total processing time of the jobs would be more than $W$.

If, in the optimal schedule, there is a huge job scheduled with other 
jobs, we can move these jobs, one by one, to the machine with minimum load. By Lemma~\ref{cl:LPadjSched}, this process decreases the $\ell_p$ norm, contradicting that we started with an optimal schedule. 
\end{proof}

\begin{lemma}\label{cl:LPnumJobs}
In any optimal schedule $S$, there are at most $\frac1\eps$ non-small jobs scheduled on each machine.
\end{lemma}

\begin{proof}
By Lemma~\ref{cl:LPoneHugeJob}, in an optimal schedule, any machine with a huge job will have only one job.

In $S$, let machine $i$ be the machine with the minimum load and assume that some machine $j \ne i$ has more than $\frac1\eps$ large jobs.  The load of $j$ is at least $(1+\eps)\frac{W}{m}$ and, hence, the load of $i$ is strictly less than $\frac{W}{m}$.  This implies that $i$ has less than $\frac1\eps$ large jobs. By Lemma~\ref{cl:LPadjSched}, moving a large job from $j$ to $i$ will decrease the $\ell_p$ norm, contradicting that $S$ is an optimal schedule. 
\end{proof}

From Lemma \ref{cl:LPnumJobs}, $v = \frac1\eps$.  Using this value with Fact \ref{fact:schGenAdv} of the general framework, gives the following.
\begin{fact}
The online algorithm with advice, based on the general framework, uses at most
 $\frac{2}{\eps}\left(\log\left(3\frac{\log(1/\eps)}{\log(1+\eps)}\right)\right) + 4$
bits of advice per request.
\end{fact}

\begin{theorem}\label{thm:LPratio}
Given a request sequence $\sigma$, $U = \frac{W}{m}$ and an $\eps$, $0 < \eps < 1/2$, the general framework  schedules the jobs of $\sigma$ such that the resulting schedule has an $\ell_p$ norm of at most $(1 + 2\eps)\OPT$.
\end{theorem}

\begin{proof}
By Lemma~\ref{cl:LPoneHugeJob}, Assumption~\ref{ass:hugeJob} holds and Theorem~\ref{thm:genFrame} applies.  

The algorithm schedules the jobs such that
\begin{align*}
\|L(\Salg)\|_p &= \left(\sum_{i=1}^m\left(L_i(\Salg)\right)^p\right)^{1/p} \\
            &\le \left(\sum_{i=1}^m\left((1+\eps)L_i(\Sstar) + \eps\frac{W}{m}\right)^p\right)^{1/p} & \text{by Theorem~\ref{thm:genFrame}} \\
            &\le \left(\sum_{i=1}^m \left((1+\eps)L_i(\Sstar)\right)^p\right)^{1/p} + \left(\sum_{i=1}^m\left(\eps\frac{W}{m}\right)^p\right)^{1/p}\\
            &\le (1+\eps)\OPT + \eps \OPT & \text{by Corollary~\ref{cor:avgIsOpt}}\\
            &= (1+2\eps)\OPT ~,
\end{align*}
where we have used the Minkowski inequality in the third line.
\end{proof}

\section{Lower Bound for Scheduling}\label{sec:lb}

Boyar et al.~\cite{BoyarKLL14} showed that at least $(n - 2N)\log N$ bits of advice in total 
(i.e.\ at least $\left(1 - \frac{2N}{n}\right)\log N$ bits per request)  
are needed for any online bin packing algorithm with advice to be optimal.  
Using a similar technique, we show that $(n - 2m)\log m$ bits of advice in total (at least $\left(1 - \frac{2m}{n}\right)\log m$ bits of advice per request) are required for any online scheduling algorithm with advice on $m$ identical machines to be optimal for makespan, machine cover or the $\ell_p$ norm. 

Let 
\begin{align*}
  k &= n - 2m, \\
  \sigma_1 &= \left<\half{k+2},\half{k+3},\ldots,\half{k+m+1},\half{2},\ldots,\half{k+1}\right> ~\text{and} \\
  \sigma_2 &= \left<x_1,x_2,\ldots,x_m\right> ~, 
\end{align*}
where $x_i$ will be defined later in an adversarial manner.  The entire adversarial request sequence will be $\sigma = \left<\sigma_1,\sigma_2\right>$.  This sequence will be chosen such that the adversary will have a balanced schedule (a load of $1$ on each machine) while any algorithm using less than $k\log m$ bits of advice will not. That is,
such an algorithm will have at least one machine with load greater than $1$, and, hence, at least one machine with load less than $1$. Such an algorithm will, therefore, not be optimal for makespan, machine cover or the $\ell_p$ norm.

\begin{fact}\label{fact:uniqueSum}
Every subset of the requests of $\sigma_1$ has a unique sum that is less than $1/2$.
\end{fact}

Let $T$ be the set of all possible schedules on $m$ identical machines for the requests of $\sigma_1$.  The adversary will schedule each of the first $m$ requests of $\sigma_1$ on a distinct machine. This distinguishes the $m$ machines from one another. Let $V$ be the set of all possible schedules of the last $k$ requests of $\sigma_1$ onto the $m$ machines, given that the first $m$ requests of $\sigma_1$ were each scheduled on a distinct machine. Note that $V \subset T$ and that $|V| = m^k$.  Let $S^\ADV_{\sigma_1} \in V$ be the adversarial schedule of the jobs of $\sigma_1$. 
Define $x_i = 1 - L_i(S^\ADV_{\sigma_1})$. Note that using Fact~\ref{fact:uniqueSum} we have that the $m$ values $x_i$, $1 \leq i \leq m$, are distinct. Further, note that $\sigma$ allows for a balanced schedule, where all machines have load $1$. 

\begin{observation}\label{obs:diffSched} 
For every $S_{\sigma_1} \in T \setminus S^\ADV_{\sigma_1}$, every possible scheduling of the jobs of $\sigma_2$ into $S_{\sigma_1}$ results in a schedule $S_{\sigma}$ such that there are at least $2$ machines $i$ and $j$, where $L_i(S_{\sigma}) < 1$ and $L_j(S_{\sigma}) > 1$.
\end{observation}
 
\begin{proof}
The sum of the processing times of all jobs of $\sigma_1$ is less than $1/2$ which implies that the processing time for each $x_i$ is greater than $1/2$.  Therefore, any machine that schedules more than one job from $\sigma_2$ will have a load greater than $1$. It follows that such a schedule also has a machine that does not have any job from $\sigma_2$, and, hence, has a load less than $1$.  We therefore consider a schedule $S_{\sigma}$ that schedules a single job from  $\sigma_2$ on each machine. 

Since the sum of the processing times of all the jobs of $\sigma$ is $m$, note that if we have a machine with a load greater than $1$ then there must be a machine with load less than $1$. We can therefore assume by contradiction that in $S_{\sigma}$ all machines have a load exactly $1$.  As each job $x_i$ of $\sigma_2$ is scheduled on a distinct machine, we have that in $S_{\sigma}$ the total processing time of the jobs from $\sigma_1$ on the machine that has job $x_i$ is exactly $1-x_i$.  Fact~\ref{fact:uniqueSum} implies that $S_{\sigma_1}$ equals $S^\ADV_{\sigma_1}$, a contradiction. 
\end{proof}

We are now ready to prove the main theorem of the section.

\begin{theorem}\label{thm:schdLB}
Any online algorithm with advice needs at least $(n - 2m)\log m$ bits of advice in order to be optimal for the makespan problem, machine cover problem and the $\ell_p$ norm problem, where $m$ is the number of machines and $n$ is the length of the request sequence.  
\end{theorem}

\begin{proof} 
Let $ALG$ be an arbitrary (deterministic) online algorithm with advice for the given scheduling problem.  Let $\Salg_{\sigma_1}$ be the schedule produced by $ALG$ for $\sigma_1$.  If $\Salg_{\sigma_1} \in T \setminus V$, i.e.\ $\Salg_{\sigma_1}$ is such that the first $m$ requests are not scheduled on distinct machines, then $\Salg_{\sigma_1} \ne S^\ADV_{\sigma_1}$, and, by Observation~\ref{obs:diffSched}, $\Salg_{\sigma}$ is not balanced.  Therefore, we will assume that the algorithm will schedule the first $m$ requests on $m$ distinct machines, i.e.\ $\Salg_{\sigma_1} \in V$.

Assume that the online algorithm with advice receives all the advice bits in advance. This only strengthens the algorithm and, thus, strengthens our lower bound.  Let $\ALG(s, u)$ be the schedule produced by $ALG$ for request sequence $s$ when receiving advice bits $u$.  Since $ALG$ gets less than $k\log m$ bits of advice, it gets as advice some $u \in U$ for some advice space $U$, $|U| < m^k$.  It follows that $|\{\ALG(\sigma_1, u) | u \in U \}| < m^k = |V|$.  Therefore, given $ALG$, $S^\ADV_{\sigma_1}$ is chosen by the adversary such that $S^\ADV_{\sigma_1} \in T \setminus \{\ALG(\sigma_1, u) | u \in U\}$.  Note that this choice defines $\sigma_2$.
 
We now have, by Observation~\ref{obs:diffSched}, that $\Salg_{\sigma}$ has at least $2$ machines $i$ and $j$ such that $L_i(\Salg_{\sigma}) < 1$ and $L_j(\Salg_{\sigma}) > 1$.  Given that there is a balanced schedule with all machines having load $1$ for $\sigma$, $\Salg_{\sigma}$ is not optimal for makespan due to machine $j$, $\Salg_{\sigma}$ is not optimal for machine cover due to machine $i$, and $\Salg_{\sigma}$ is not optimal for the $\ell_p$ norm by Corollary~\ref{cor:avgIsOpt}.
\end{proof}

\section{Comparison to the Semi-Online Advice Model}

For a request sequence of length $n$, the na\"ive conversion of the algorithms described previously from the online advice model to the semi-online advice  model uses less than a total of $n \left(\bpAdvUp\right)$ bits of advice for bin packing and $n \left( \log\left(3\log_{1+\eps}1/\eps\right) + \beta(v) + 3 \right)$ for scheduling. It is possible, as we describe below, to do better in the more powerful semi-online model, but the amount of advice is still linear in $n$. This follows from the observation that only for the first $N$ ($m$) request does the advice include an actual bin (machine) pattern.

\subsection{$\BPA$ in the semi-online advice model}

Initially, a single bit $w$ (as described above) is written to the advice tape to indicate if $N \le 1/\eps$ or not. If so, $n\log(1/\eps)$ bits are written to the tape to indicate the bin index in which to pack each item.  

If $N > 1/\eps$, the optimal number of bins $N$ is encoded, using a self-delimiting encoding scheme~\cite{BockKKR14}, and written to the advice tape, using $\lceil \log N \rceil + 2 \lceil \log \lceil \log N \rceil \rceil$ bits. Then, for each of the $N$ optimal bins, the bin pattern is written, using  $p< (1/\eps)\log(2/\eps^2) + 1$ bits, followed by a bit to indicate if small items are packed in the bin. 

For each request, a bit is written to the advice tape to indicate if the item is small or large. If the requested item is small, an additional bit is written to indicate if the small item should be packed in the current bin packing small items or the next. If the requested item is large, the item type is written, using $t < \log(2/\eps^2) + 1$ bits, and an additional bit is written to indicate if the large item is packed in a bin with or without small items. 

The total amount of advice used is less than $$1 + \lceil \log N \rceil + 2 \lceil \log \lceil \log N \rceil \rceil + N\left(\frac{1}{\eps}\log\left(\frac{2}{\eps^2}\right) + 1\right) + N + n\left(\log\left(\frac{2}{\eps^2}\right) + 2\right)~.$$

In \cite{AngelopoulosDKRR2015}, it is shown that a linear amount of advice in total is required for an algorithm with advice for the bin packing problem to achieve a competitive ratio of $7/6$. Therefore, the algorithms described here uses an optimal amount of advice (up to constant factors) for competitive ratios at most $7/6$. 

\subsection{Scheduling framework in the semi-online advice model}

As with $\BPA$, the machine pattern of each machine and the machines that have small items scheduled can be written to the front of the advice tape. In this case, $m$ unlike $N$ is known to algorithm and does not need to be written to the advice tape.   

Initially, the $m$ machine patterns, ordered according to the permutation of the machines as described previously, are written on the advice tape using $p<m\beta(v)$ bits. Then, for each request, the type of each job is written to the advice tape, using $t<\log\left(3\log(1/\eps)/\log(1 + \eps)\right) + 1$ bits. If the job is small, then an additional bit is written to indicate if the small job should be scheduled on the previous machine or the current machine.

The total amount of advice used is less than $$m\beta(v) + n\left( \log\left(\frac{3\log(1/\eps)}{\log(1 + \eps)}\right) + 2\right)~.$$

The framework that we presented here works for any objective function of the form $\sum_{i=1}^m f(L_i)$ such that if $x\le (1+\eps)y$ then $f(x)\le (1+O(1)\eps)f(y)$ (which includes makespan, machine cover and the $\ell_p$ norm). In \cite{Dohrau2015}, an algorithm that uses constant advice in total and achieves a competitive ratio of $1+\eps$ is presented for the makespan objective on $m$ identical machines. It remains open whether or not it is possible to improve on the framework here for other objective functions such as machine cover and the $\ell_p$ norm.

\section{Conclusions}\label{sec:con}
We gave online algorithms with advice for bin packing and scheduling problems 
that, with a constant number of bits per request, achieve competitive ratios arbitrarily close to $1$.
Since this is not possible for all online problems, it would be interesting to prove similar results for additional online problems. Furthermore, an interesting question is to find the right trade-off between the (constant) number of bits of advice and the achievable competitive ratios for the problems we study and other problems. 

\paragraph{Acknowledgements}
We would like to thank Shahin Kamali and Alejandro L{\'o}pez-Ortiz for useful discussions about the bin packing problem. 

\bibliographystyle{alpha}
\bibliography{advice}

\newcommand{\etalchar}[1]{$^{#1}$}
\begin{thebibliography}{ARRvS13}

\bibitem[AAS98]{AvidorAS1998}
Adi Avidor, Yossi Azar, and Ji\v{r}\'{\i} Sgall.
\newblock Ancient and new algorithms for load balancing in the lp norm.
\newblock In {\em SODA}, pages 426--435, 1998.

\bibitem[AAWY97]{Alon1997}
N.~Alon, Y.~Azar, G.J. Woeginger, and T.~Yadid.
\newblock Approximation schemes for scheduling.
\newblock In {\em SODA}, pages 493--500, 1997.

\bibitem[ADK{\etalchar{+}}ar]{AngelopoulosDKRR2015}
Spyros Angelopoulos, Christoph D{\"{u}}rr, Shahin Kamali, Marc Renault, and Adi
  Ros{\'{e}}n.
\newblock Online bin packing with advice of small size.
\newblock In {\em Algorithms and Data Structures - 14th International
  Symposium, {WADS} 2015, Victoria, BC, Canada, August 5-7, 2015. Proceedings},
  pages 1--12, to appear.

\bibitem[AE98]{AzarE1998}
Yossi Azar and Leah Epstein.
\newblock On-line machine covering.
\newblock {\em J Scheduling}, 1(2):67--77, 1998.

\bibitem[Alb02]{Albers2002}
Susanne Albers.
\newblock On randomized online scheduling.
\newblock In {\em STOC}, pages 134--143. ACM, 2002.

\bibitem[ARRvS13]{AdamaszekRRvS2013}
Anna Adamaszek, Marc~P. Renault, Adi Ros{\'{e}}n, and Rob van Stee.
\newblock Reordering buffer management with advice.
\newblock In Christos Kaklamanis and Kirk Pruhs, editors, {\em Approximation
  and Online Algorithms - 11th International Workshop, {WAOA} 2013, Sophia
  Antipolis, France, September 5-6, 2013, Revised Selected Papers}, volume 8447
  of {\em Lecture Notes in Computer Science}, pages 132--143. Springer, 2013.

\bibitem[BBG12]{BaloghBG2012}
J\'{a}nos Balogh, J\'{o}zsef B{\'e}k{\'e}si, and G\'{a}bor Galambos.
\newblock New lower bounds for certain classes of bin packing algorithms.
\newblock {\em TCS}, 440-441(0):1--13, 2012.

\bibitem[BEY98]{BorodinBook1998}
Allan Borodin and Ran El-Yaniv.
\newblock {\em Online computation and competitive analysis}.
\newblock Cambridge University Press, New York, NY, USA, 1998.

\bibitem[BKK{\etalchar{+}}09]{BockenhauerKKKM09}
Hans-Joachim B{\"o}ckenhauer, Dennis Komm, Rastislav Kr{\'a}lovic, Richard
  Kr{\'a}lovic, and Tobias M{\"o}mke.
\newblock On the advice complexity of online problems.
\newblock In {\em ISAAC}, pages 331--340, 2009.

\bibitem[BKKK11]{BockenhauerKKK11}
Hans-Joachim B{\"o}ckenhauer, Dennis Komm, Rastislav Kr{\'a}lovic, and Richard
  Kr{\'a}lovic.
\newblock On the advice complexity of the k-server problem.
\newblock In {\em ICALP}, pages 207--218, 2011.

\bibitem[BKKR14]{BockKKR14}
Hans{-}Joachim B{\"{o}}ckenhauer, Dennis Komm, Richard Kr{\'{a}}lovic, and
  Peter Rossmanith.
\newblock The online knapsack problem: Advice and randomization.
\newblock {\em Theor. Comput. Sci.}, 527:61--72, 2014.

\bibitem[BKLL14]{BoyarKLL14}
Joan Boyar, Shahin Kamali, Kim~S. Larsen, and Alejandro L{\'{o}}pez{-}Ortiz.
\newblock Online bin packing with advice.
\newblock In Ernst~W. Mayr and Natacha Portier, editors, {\em 31st
  International Symposium on Theoretical Aspects of Computer Science {(STACS}
  2014), {STACS} 2014, March 5-8, 2014, Lyon, France}, volume~25 of {\em
  LIPIcs}, pages 174--186. Schloss Dagstuhl - Leibniz-Zentrum fuer Informatik,
  2014.

\bibitem[Cha92]{Chandra92}
Barun Chandra.
\newblock Does randomization help in on-line bin packing?
\newblock {\em Inf. Process. Lett.}, 43(1):15--19, August 1992.

\bibitem[CvVW94]{ChenVW1994}
Bo~Chen, Andr\'{e} van Vliet, and Gerhard~J. Woeginger.
\newblock A lower bound for randomized on-line scheduling algorithms.
\newblock {\em Information Processing Letters}, 51(5):219 -- 222, 1994.

\bibitem[DHZ12]{DorrigivHZ2012}
Reza Dorrigiv, Meng He, and Norbert Zeh.
\newblock On the advice complexity of buffer management.
\newblock In {\em ISAAC}, pages 136--145, 2012.

\bibitem[DLO05]{DorrigivL2005}
Reza Dorrigiv and Alejandro L{\'o}pez-Ortiz.
\newblock A survey of performance measures for on-line algorithms.
\newblock {\em SIGACT News}, 36(3):67--81, 2005.

\bibitem[Doh15]{Dohrau2015}
J{\'{e}}r{\^{o}}me Dohrau.
\newblock Online makespan scheduling with sublinear advice.
\newblock In Giuseppe~F. Italiano, Tiziana Margaria{-}Steffen, Jaroslav
  Pokorn{\'{y}}, Jean{-}Jacques Quisquater, and Roger Wattenhofer, editors,
  {\em {SOFSEM} 2015: Theory and Practice of Computer Science - 41st
  International Conference on Current Trends in Theory and Practice of Computer
  Science, Pec pod Sn{\v{e}}{\v{z}}kou, Czech Republic, January 24-29, 2015.
  Proceedings}, volume 8939 of {\em Lecture Notes in Computer Science}, pages
  177--188. Springer, 2015.

\bibitem[EFKR11]{EmekFKR11}
Yuval Emek, Pierre Fraigniaud, Amos Korman, and Adi Ros{\'e}n.
\newblock Online computation with advice.
\newblock {\em Theor. Comput. Sci.}, 412(24):2642--2656, 2011.

\bibitem[FdlVL81]{VegaL1981}
Wenceslas Fernandez de~la Vega and George~S. Lueker.
\newblock Bin packing can be solved within 1+epsilon in linear time.
\newblock {\em Combinatorica}, 1(4):349--355, 1981.

\bibitem[FW00]{FleischerW2000}
Rudolf Fleischer and Michaela Wahl.
\newblock Online scheduling revisited.
\newblock In {\em ESA}, pages 202--210, 2000.

\bibitem[GKLO13]{GuptaKL13}
Sushmita Gupta, Shahin Kamali, and Alejandro L{\'o}pez-Ortiz.
\newblock On advice complexity of the k-server problem under sparse metrics.
\newblock In Thomas Moscibroda and Adele~A. Rescigno, editors, {\em SIROCCO},
  volume 8179 of {\em Lecture Notes in Computer Science}, pages 55--67.
  Springer, 2013.

\bibitem[HS87]{HochbaumS1987}
Dorit~S. Hochbaum and David~B. Shmoys.
\newblock Using dual approximation algorithms for scheduling problems
  theoretical and practical results.
\newblock {\em J. ACM}, 34(1):144--162, January 1987.

\bibitem[Joh73]{Johnson1973}
D.S. Johnson.
\newblock {\em Near-optimal Bin Packing Algorithms}.
\newblock PhD thesis, MIT, 1973.

\bibitem[KK11]{KommK11}
Dennis Komm and Richard Kr{\'a}lovic.
\newblock Advice complexity and barely random algorithms.
\newblock {\em RAIRO - Theor. Inf. and Applic.}, 45(2):249--267, 2011.

\bibitem[RC03]{RudinC2003}
John~F. Rudin, III and R.~Chandrasekaran.
\newblock Improved bounds for the online scheduling problem.
\newblock {\em SIAM J. Comput.}, 32(3):717--735, March 2003.

\bibitem[RR15]{RenaultR2015}
Marc~P. Renault and Adi Ros{\'{e}}n.
\newblock On online algorithms with advice for the k-server problem.
\newblock {\em Theory Comput. Syst.}, 56(1):3--21, 2015.

\bibitem[Sei01]{Seiden2001}
Steven~S. Seiden.
\newblock On the online bin packing problem.
\newblock {\em J. ACM}, 49:2002, 2001.

\bibitem[Sga97]{Sgall1997}
Ji\v{r}\'{\i} Sgall.
\newblock A lower bound for randomized on-line multiprocessor scheduling.
\newblock {\em Inf. Process. Lett.}, 63(1):51 -- 55, 1997.

\bibitem[Woe97]{Woeginger1997}
Gerhard~J. Woeginger.
\newblock A polynomial-time approximation scheme for maximizing the minimum
  machine completion time.
\newblock {\em Oper. Res. Lett.}, 20(4):149 -- 154, 1997.

\end{thebibliography}

\end{document}